\documentclass{article}
\usepackage{iclr2027_conference,times}

\usepackage{booktabs}
\usepackage{hyperref}
\usepackage{url}
\usepackage{float}
\usepackage{graphicx}
\usepackage{xspace}
\usepackage{wrapfig}
\usepackage[section]{placeins}

\usepackage{amsthm}
\newtheorem{theorem}{Theorem}
\newtheorem{proposition}[theorem]{Proposition}
\newtheorem{lemma}[theorem]{Lemma}
\newtheorem{corollary}[theorem]{Corollary}
\theoremstyle{definition}
\newtheorem{assumption}[theorem]{Assumption}

\usepackage{amsmath,amsfonts,bm}

\def\eqref#1{equation~\ref{#1}}

\def\1{\bm{1}}

\DeclareMathAlphabet{\mathsfit}{\encodingdefault}{\sfdefault}{m}{sl}
\SetMathAlphabet{\mathsfit}{bold}{\encodingdefault}{\sfdefault}{bx}{n}

\newcommand{\method}{\textbf{\texttt{ReGap}}\xspace}

\title{
Sleeping Secrets: How Fine-Tuning Reawakens Privacy Risks in Language Models
}

\iclrfinalcopy

\author{
Jianhong Li\textsuperscript{\normalfont 1}\thanks{Equal contribution.},
Jiahao Chen\textsuperscript{\normalfont 2}\footnotemark[1],
Yuwen Pu\textsuperscript{\normalfont 1}\thanks{Corresponding authors.}, 
Chunyi Zhou\textsuperscript{\normalfont 1},
Oubo Ma\textsuperscript{\normalfont 1},
Zhou Feng\textsuperscript{\normalfont 1},\And
Hangtao Zhang\textsuperscript{\normalfont 3}, 
Jichao Bi\textsuperscript{\normalfont 1}, 
Chunqiang Hu\textsuperscript{\normalfont 1}\\
\\
\textsuperscript{1}Chongqing University, \quad \\
\textsuperscript{2}Zhejiang University \quad \\
\textsuperscript{3}Huazhong University of Science and Technology
}

\begin{document}

\maketitle

\fancyhead{}
\renewcommand{\headrulewidth}{0pt}

\begin{abstract}
Beyond adapting Large Language Models (LLMs) to specialized applications, fine-tuning has recently been shown to recover private information that is no longer accessible through direct queries. Previous fine-tuning recovery attacks, however, require genuine private supervision drawn from the same distribution, i.e., the previous training dataset. We argue that such recovery remains possible without such impractical knowledge. We show that LLM-generated candidates can provide sufficient supervision to recover previously learned private associations. Based on this, we propose \method, a data-free attack that recovers private associations using task structure, filters them by answer-token likelihood, and updates the target model via low-rank adaptation. Specifically, \method requires neither target answers nor auxiliary genuine private supervision. Across six GPT-2, OPT, and Qwen3 models, \method improves target-association recovery by 6-21 percentage points over the post-training target model. Recovery remains substantial even when the adaptation identities are disjoint from all memorized and evaluation identities, with no exact target answers appearing in the generated or selected supervision. Moreover, the same trained adapters increase recovery from 42\% to 63\% on a previously exposed checkpoint, but produce no gain on a matched checkpoint that never encountered the targets. This contrast shows that adaptation alone is insufficient to explain the observed recovery and that prior target exposure strongly affects post-adaptation recoverability. Our findings highlight that routine model customization can reawaken latent privacy risks, warranting urgent attention from the academic and industrial communities.
\end{abstract}

\section{Introduction}

Fine-tuning has become a routine part of the lifecycle of Large Language Models (LLMs), enabling pretrained models to adapt to new tasks, domains, and applications \citep{gururangan2020don,hu2021lora}. This continued adaptation also has implications for privacy. LLMs can memorize sensitive information and personally identifiable information (PII) from their pretraining data and disclose it in response to carefully crafted queries~\citep{carlini2019secret,carlini2021extracting,tirumala2022memorization,
biderman2023emergent,lukas2023analyzing}. Recent work even shows that adaptation can introduce additional privacy risks: sensitive information may be
unintentionally memorized during fine-tuning, model-generated training data
can amplify privacy leakage, and even benign fine-tuning can weaken contextual
privacy protections \citep{szep2026unintended,akkus2025generated,
goel2026privacy}. These findings motivate evaluating privacy not only at a
fixed checkpoint, but across the model's subsequent adaptation lifecycle.

A particularly relevant line of work studies whether later updates can restore
access to information that is difficult to elicit directly. Janus shows that
fine-tuning on genuine private association pairs can recover previously learned
PII \citep{chen2024janus}. Related studies show that downstream fine-tuning can
also reverse apparent forgetting after machine unlearning
\citep{wang2025invariance}, while malicious model-side interventions can
increase privacy leakage during subsequent fine-tuning
\citep{rashid2025forget}. These results establish post-adaptation privacy as
an important attack surface, but they leave a different attacker capability
largely unexplored. Specifically, they assume access to authentic private records from
the target training distribution~\cite{chen2024janus}, an advantage that may be unavailable to an
attacker who only obtains the released model. In this setting, the attacker has access to the released post-training model, but not to the original model from which it was derived. This released model
is our target model: the attacker seeks to recover private associations
learned before post-training, without genuine private records as supervision.

This leads to our central question: can previously learned associations become more recoverable using only model-generated supervision, without externally supplied private records? We further ask whether recovery persists when adaptation uses neither target identities nor exact target answers. To investigate these questions, we propose \method, a recovery attack that generates candidate
associations from the target model, selects candidates using answer-token
likelihood under the same frozen model, and trains a fresh LoRA adapter on the
resulting model-generated supervision. \method uses neither target answers nor
auxiliary genuine private records during candidate generation, selection, or
adaptation.

Across six GPT-2, OPT, and Qwen3 settings, one round of \method improves
target-association recovery by 6-21 percentage points over the target model. Recovery remains substantial when adaptation identities are
disjoint from all memorized and evaluation identities, with no exact target
answer appearing in either the generated candidate pools or selected examples.
Moreover, applying the same trained adapters to matched checkpoints yields a
21-point gain only when the model previously encountered the targets, while a
target-unexposed checkpoint shows no gain. These controls indicate that the
observed recovery depends on prior exposure rather than adaptation alone,
while avoiding any claim that individual memories remain internally intact.

Our main contributions are summarized as follows:

\begin{itemize}

\item We identify a post-adaptation privacy risk in which previously learned
private associations can become substantially more recoverable even without
genuine private supervision. In the structured association setting studied
here, recovery can persist when adaptation data are constructed without
target identities and without exact target-answer replay, showing that direct
behavioral inaccessibility alone is insufficient to characterize later privacy
exposure.

\item We introduce \method, a recovery attack that replaces genuine private
supervision with model-generated candidates. \method follows a
Generate--Select--Adapt pipeline: it generates task-consistent candidates from
the target model, ranks them by answer-token likelihood under the frozen model,
and trains a fresh LoRA adapter on the selected supervision. We further provide
a local gradient-based analysis that characterizes when auxiliary updates can
reduce target loss without identity or answer overlap and motivates the
frozen-model selection rule.

\item We systematically characterize the effect and its boundaries across
GPT-2, OPT, and Qwen3 models. One round of \method improves
target-association recovery by 6--21 percentage points across six model
settings; matched controls further show substantial recovery with unseen
adaptation identities and a strong dependence on prior target exposure.
At the same time, recovery weakens sharply for opaque random bindings and
after explicit NPO unlearning, establishing important limits on when
model-generated supervision enables recovery.

\end{itemize}

\section{Related Work}
\label{sec:related}

\subsection{Memorization and Behavioral Accessibility}
Language models can memorize specific training examples, with memorization
affected by repetition, model scale, and training dynamics
\citep{carlini2019secret,tirumala2022memorization,carlini2022quantifying,
biderman2023emergent}. Memorization also extends beyond verbatim sequences to
entity-attribute associations that are directly relevant to privacy
\citep{zhou2024quantifying,chen2026profiling,li2026auditing}. However, whether previously learned information can be elicited from a checkpoint may change after subsequent updates
\citep{jang2022,jagielski2022measuring}.

Recent work further highlights the distinction between behavioral forgetting
and information removal. Latent-memory analyses and unlearning audits find
evidence that information can remain recoverable or decodable even after its
direct accessibility has been substantially reduced
\citep{duan2025uncovering,goel2026auditing,cohen2026remind}.
We therefore use \emph{Hidden} purely as a behavioral notion: low direct
recovery does not imply either deletion or intact internal storage.

\subsection{Privacy Extraction from Language Models}
Training-data extraction studies show that memorized text and PII can be elicited through prompting, prompt optimization, and large-scale querying
\citep{carlini2021extracting,lukas2023analyzing,ozdayi2023controlling,
keum2025private,nasr2025scalable}. These attacks primarily operate through
inference-time interaction with a fixed model checkpoint.

Importantly, successful prediction of a private attribute does not necessarily
imply item-specific memorization, since lexical and distributional cues can also
support reconstruction \citep{luo2026llms,pu2025facial}. This distinction is particularly
relevant to our setting, where structured associations may contain exploitable
regularities even when the corresponding target answers are unavailable to the
attacker.

\subsection{Post-Adaptation Privacy Risks}
Privacy risk can change after a model undergoes subsequent adaptation.
Post-training updates can substantially alter model behavior
\citep{gururangan2020don,qi2024fine}. This extends to privacy:
fine-tuning on generated data can increase PII extraction and membership
inference risk \citep{akkus2025generated}, benign adaptation can degrade
contextual privacy \citep{goel2026privacy}, and changes involving one
set of personal records can affect behavior on related information
\citep{borkar2025privacy}.

Most closely related to our work, Janus moves beyond fixed-checkpoint
extraction by using fine-tuning as the attack interface: genuine private
association pairs are used to restore access to previously learned PII
\citep{chen2024janus}. A related line of work studies the interaction between
forgetting and subsequent adaptation. Downstream fine-tuning can recover
information targeted by machine unlearning, motivating unlearning procedures
designed to remain invariant to later updates
\citep{wang2025invariance}. Conversely, Rashid et al.\ show that a maliciously
modified pretrained model can exploit unlearning to amplify privacy leakage
during later victim fine-tuning \citep{rashid2025forget}.

\method studies a different supervision regime. It neither designs an unlearning
defense nor modifies an original model to implant a vulnerability.
Instead, the attacker starts from a given target model in which previously learned associations have low direct recoverability and
constructs supervision from the model itself, without access to target answers
or genuine private association pairs. Our question is therefore not simply
whether adaptation can alter privacy, but whether previously learned
associations can become more accessible again when genuine private supervision
is unavailable.

\section{Problem Formulation}
\label{sec:problem}

We focus on whether fine-tuning can make previously learned private associations
accessible again after their direct recoverability has declined. Let
$\mathcal D_{\mathrm{tar}}=\{(x_i,y_i)\}_{i=1}^{N}$, where $N\ge 1$,
$x_i$ is a query key, and $y_i$ is a nonempty private attribute. A fixed task
template $\tau$ gives $q_i=\tau(x_i)$ and $t_i=(q_i,y_i)$. Under a fixed
evaluation protocol, let $\hat y_i(M)$ be the continuation generated by $M$
for $q_i$, excluding the prompt. We measure recovery by
\begin{equation}
\operatorname{Rec}(M)=\frac{1}{N}\sum_{i=1}^{N}
\mathbf 1\!\left[
\operatorname{Norm}(y_i)\ne\varepsilon
\;\wedge\;
\operatorname{Norm}(y_i)
\preceq_{\mathrm{sub}}\operatorname{Norm}(\hat y_i(M))
\right],
\label{eq:recovery}
\end{equation}
where $\epsilon$ is the empty string, $\preceq_{\mathrm{sub}}$ denotes
contiguous substring inclusion, and $\mathrm{Norm}$ lowercases text, trims
surrounding whitespace, and collapses consecutive whitespace. Prompting, decoding settings, and generation budgets are held fixed across models; any evaluation randomness is included in the evaluation protocol.
Let $M_{\mathrm{m}}$ be an original model that learned the target
associations. Subsequent non-target training produces a post-training model $M_h$
with $\operatorname{Rec}(M_h)<\operatorname{Rec}(M_{\mathrm{m}})$~\citep{chen2024janus}. This establishes neither deletion nor intact internal retention of the associations. The attacker receives only $M_h$.

\paragraph{Threat model.}
The attacker knows the association schema and answer format and can query
the target model $M_h$, evaluate token-level likelihoods, and fine-tune
it. Adaptation uses auxiliary keys $X_{\mathrm{aux}}$ and model-generated
answers, without externally supplied target answers or genuine private
association pairs. No distributional match between externally supplied
auxiliary data and the targets is required. We use \emph{data-free} in this
supervision sense; generated candidates may nevertheless coincide with true
answers. Correctness labels, target-recovery feedback, and $M_{\mathrm{m}}$
are unavailable for attack construction or model selection. Ground truth is
reserved for evaluation and post-hoc analysis.
In the target-aware setting, $X_{\mathrm{aux}}$ may include known target
keys. In the \emph{Unseen-Identity Control}, it is disjoint from the keys
used in the controlled memorization stage and target evaluation. Target
identities are therefore unnecessary for constructing adaptation data,
although evaluation queries still use the target keys.

\paragraph{Objective.}
Let $\mathcal D_{\mathrm{aux}}$ denote the auxiliary fine-tuning dataset
constructed by selecting model-generated QA pairs. With $\mathcal K_{\mathrm{att}}$ denoting the task knowledge and auxiliary keys, attacker constructs 
\begin{equation}
    \mathcal D_{\mathrm{aux}}
=\mathcal B(M_h,\mathcal K_{\mathrm{att}}),\quad M_a=\mathcal A(M_h,\mathcal D_{\mathrm{aux}}).
\end{equation}
Our objective is to determine whether the attacker can achieve 
\begin{equation}
    \Delta_{\mathrm{rec}}=\operatorname{Rec}(M_a)-\operatorname{Rec}(M_h)>0.
\end{equation}

\section{From Recovery Geometry to \method}
\label{sec:method}

\method{} constructs fine-tuning supervision from the target model itself through a Generate-Select-Adapt pipeline, as illustrated in Figure~\ref{fig:regap}. We first characterize how an auxiliary update can improve unknown targets, then motivate an observable selection rule and specify the Generate-Select-Adapt procedure.
\begin{figure}[t]
    \centering
    \includegraphics[width=0.9\linewidth]{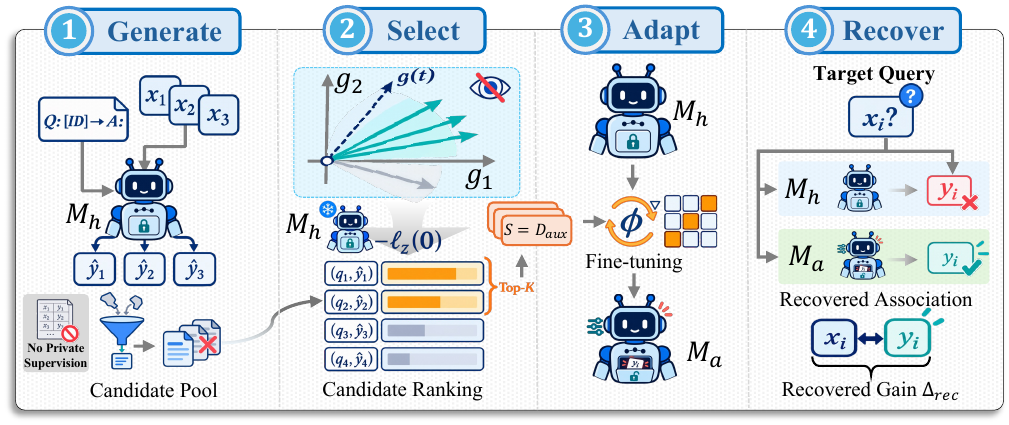}
    \vspace{-10pt}
    \caption{Overview of \method. Generate--Select--Adapt for recovering hidden associations.}
    \label{fig:regap}
\end{figure}

\subsection{Local Geometry of Fine-Tuning}
\label{sec:recovery_geometry}

Let $M(\phi)$ be a parameterization of fine-tuning, with
$\phi\in\mathbb R^d$ denoting displacements in the chosen trainable
coordinates and $M(0)=M_h$. For $z=(q,y)$ with $|y|\ge 1$ answer tokens,
define
\begin{equation}
\ell_z(\phi)
=-\frac{1}{|y|}\sum_{t=1}^{|y|}
\log p_{M(\phi)}(y_t\mid q,y_{<t}),
\qquad
g(z)=\nabla_\phi\ell_z(0).
\label{eq:answer_loss}
\end{equation}
Prompt tokens do not contribute to the loss. For a nonempty auxiliary set
$S$, let 
\begin{equation}
    L_T(\phi)=\frac{1}{N}\sum_{i=1}^{N}\ell_{t_i}(\phi), \quad L_S(\phi)=\frac{1}{|S|}\sum_{z\in S}\ell_z(\phi), 
\end{equation}
and write $g_T=\nabla L_T(0)$ and $g_S=\nabla L_S(0)$. The local gradient-alignment kernel is
\begin{equation}
\kappa_h(z,z')=\langle g(z),g(z')\rangle.
\label{eq:recovery_kernel}
\end{equation}

\begin{theorem}[Local target-loss change]
\label{thm:local_recovery}
If $L_T$ is $\beta$-smooth on the segment from $0$ to a parameter
displacement $\delta$, then 
\begin{equation}
    L_T(\delta)\le L_T(0)+\langle g_T,\delta\rangle +\frac{\beta}{2}\|\delta\|_2^2. 
\label{eq:general_update_bound}
\end{equation}

In particular, for the reference gradient step $\delta=-\eta g_S$,
$\eta>0$, smoothness on that segment gives
\begin{equation}
L_T(-\eta g_S)
\le L_T(0)
-\frac{\eta}{N|S|}\sum_{i=1}^{N}\sum_{z\in S}\kappa_h(t_i,z)
+\frac{\beta\eta^2}{2}\|g_S\|_2^2.
\label{eq:local_recovery_bound}
\end{equation}
Thus, the target loss strictly decreases if
\begin{equation}
\langle g_T,g_S\rangle
>\frac{\beta\eta}{2}\|g_S\|_2^2.
\label{eq:recovery_condition}
\end{equation}
\end{theorem}

Theorem~\ref{thm:local_recovery} separates answer overlap from update
alignment: auxiliary examples need not reproduce target answers to decrease
target loss. For a candidate distribution $Q$, the quantity $\sigma_i(Q)=\mathbb E_{z\sim Q}[\kappa_h(t_i,z)]$ measures first-order susceptibility to auxiliary gradient updates in the
chosen coordinates. Appendix~\ref{app:shared_structure} gives sufficient
shared-gradient conditions for $\sigma_i(Q)>0$ even when candidate identities
and answers do not overlap the target. This motivates the unseen-identity
control without requiring an intact internal representation of each record.
The general bound applies to an actual parameter displacement, whereas the
kernel expression specializes to a gradient step. Its trajectory-wise
extension appears in Appendix~\ref{app:multistep_bound}. A decrease in
$L_T$ alone does not imply an increase in $\operatorname{Rec}$:
Appendix~\ref{app:margin_crossing} supplies a decoding-specific sufficient
condition and distinguishes new recoveries from lost recoveries.

\subsection{From Self-Generated Answers to an Observable Selector}
\label{sec:observable_selection}

\paragraph{Why generation and selection matter.}
Self-generated supervision does not automatically produce a directed update.
Under exact on-policy sampling and a matched complete-sequence likelihood,
the expected gradient at the generating model is zero
(Proposition~\ref{prop:self_sampling_nullity}). Answer-length weighting,
biased decoding, and filtering need not preserve this cancellation.
\method{} combines greedy and temperature-controlled top-$p$ generation,
format filtering, and support-based selection to shape the auxiliary data.
These operations can create a nonzero mean update; whether that update
benefits the targets depends on its alignment, not merely on its being
nonzero.

\paragraph{An observable proxy for unavailable alignment.}
The first-order contribution of a candidate is
\begin{equation}
I_T(z)=\langle g_T,g(z)\rangle
=\frac{1}{N}\sum_{i=1}^{N}\kappa_h(t_i,z).
\label{eq:oracle_influence}
\end{equation}
This score requires unknown target answers. Instead, \method{} ranks
candidates by their support under the unchanged target model,
\begin{equation}
s_h(z)=-\ell_z(0).
\label{eq:frozen_support}
\end{equation}
Support is observable but is neither a correctness label nor an alignment
measurement. The following condition identifies when its ranking is useful.

\begin{assumption}[Conditional support--alignment ordering]
\label{ass:support_alignment}
For a random candidate pool $\mathcal C=(z_1,\ldots,z_m)$ of fixed size
$m\ge K\ge1$, let $\mathbf s=(s_h(z_1),\ldots,s_h(z_m))$ and
$a_j(\mathbf s)=\mathbb E[I_T(z_j)\mid\mathbf s]$, with integrable
$I_T(z_j)$. Almost surely, for every $j,k$, 
\begin{equation}
    s_h(z_j)\ge s_h(z_k) \quad\Longrightarrow\quad a_j(\mathbf s)\ge a_k(\mathbf s).
\end{equation}
For variable-size pools, the condition is imposed conditional on pool size.
\end{assumption}

\begin{proposition}[Conditional first-order selection]
\label{prop:topk_selector}
Under Assumption~\ref{ass:support_alignment}, selecting the $K$ largest
support scores maximizes
$\mathbb E[\langle g_T,g_S\rangle\mid\mathbf s]$ among all selectors
that choose exactly $K$ candidates using only the scores and independent
randomization.
\end{proposition}

This is a conditional justification for the first-order selection criterion,
not a guarantee that support ranks alignment in every model. Actual
target-loss improvement additionally requires positive alignment large
enough to dominate the update penalty in
Theorem~\ref{thm:local_recovery}; Appendix~\ref{app:selector_bound}
makes this condition explicit. Comparison with selection from the same pool
at random assesses the selector's practical value, while post-hoc gradient
analysis can assess its alignment premise.

\subsection{\method: Generate--Select--Adapt}
\label{sec:regap_algorithm}

\paragraph{Generate task-consistent candidates.}
For auxiliary keys $X_{\mathrm{aux}}=\{u_a\}_{a=1}^{n_{\mathrm{aux}}}$,
construct $q_a^{\mathrm{aux}}=\tau(u_a)$. Using a frozen copy of $M_h$,
generate one greedy answer and $J$ temperature-controlled top-$p$ samples
per query. Remove empty or malformed answers using only the known answer
format, and deduplicate identical normalized answers within each query.
This yields a pool $\mathcal C=\{(q_a^{\mathrm{aux}},\hat y_{a,j})\}$
of $m$ distinct query--answer candidates. No candidate is checked against
a true private attribute. The same construction applies to target-aware
and unseen-identity auxiliary keys.

\paragraph{Select by frozen support.}
For every $z=(q,\hat y)\in\mathcal C$, compute
\begin{equation}
\mathcal L_{\mathrm{sel}}(z)
=-\frac{1}{|\hat y|}\sum_{t=1}^{|\hat y|}
\log p_{M_h}(\hat y_t\mid q,\hat y_{<t})
=\ell_z(0).
\label{eq:selection_loss}
\end{equation}
With a prescribed budget $K_0\ge1$, retain the
$K=\min(K_0,m)$ candidates with smallest selection loss, breaking ties
by a fixed rule. Selection is global across the pool and uses the target
model likelihoods, not the temperature- or top-$p$-modified generation
probabilities. If the filtered pool is empty, return $M_h$ without adaptation. The selected set $S$ constitutes the adaptation dataset, i.e., $\mathcal D_{\mathrm{aux}}=S$.

\paragraph{Fine-tune from the target model.}
Starting from $\phi_0=0$, take $H$ optimization steps on the selected
answer-token objective,
\begin{equation}
L_S(\phi)=\frac{1}{K}\sum_{z\in S}\ell_z(\phi),
\qquad
M_a=M(\phi_H).
\label{eq:regap_adaptation}
\end{equation}
The loss averages tokens within each answer and then averages examples.
Generation and scoring use the unchanged target model; only the model
being adapted is updated. Generation, selection, and training budgets are
fixed without target-answer validation or target-recovery feedback.
This specifies one Generate--Select--Adapt round with multi-step
fine-tuning, rather than an additional outer self-training loop.

\section{Evaluation}

\label{sec:evaluation}

We evaluate \method along three complementary questions:
\textbf{(RQ1)} whether model-generated supervision enables recovery across
models, \textbf{(RQ2)} how much target-specific information is required, and
\textbf{(RQ3)} whether frozen-model selection contributes beyond random
candidate selection. We additionally use a matched prior-exposure control to
test whether the observed gains can be explained by ordinary relearning alone.
\subsection{Experimental Setup}
\label{sec:setup}

\paragraph{Benchmark.}
We primarily evaluate \method on 3,000 structured person-email associations
derived from the Enron email corpus \citep{klimt2004enron} using the LLM-PBE
pipeline \citep{li2024llm}. We reserve 50 associations as genuine-private
references and 100 as recovery targets. The remaining 2,850 associations form
the non-target pool, of which 300 are used for non-target evaluation.
We use WikiText-2 \citep{merity2016pointer} for subsequent non-target continual
training that produces the post-training model $M_h$. Enron person-email associations
have been used in prior language-model privacy and recovery studies
\citep{li2024llm,chen2024janus}, making them a natural testbed for our
setting. We nevertheless treat evaluation on a single primary benchmark as a
limitation and examine additional structural controls in the appendix.

\paragraph{Models.}
We evaluate six open-source autoregressive models spanning three families and
approximately 0.1B-4B parameters: GPT-2 Small, Medium, and Large
\citep{radford2019language}, OPT-1.3B~\citep{zhang2022opt}, and
Qwen3-1.7B/4B~\citep{yang2025qwen3}. These models satisfy the access required
by our threat model-token-level likelihood scoring and parameter-efficient
adaptation-while providing both controlled within-family scaling through
GPT-2 and cross-family evaluation through OPT and Qwen3.

\paragraph{Baselines.}
We use complementary baselines to isolate different sources of recovery.
\emph{Hidden} measures direct recoverability before adaptation.
\emph{Matched Random} samples the same number of examples from the same
model-generated candidate pool and uses the same LoRA adaptation budget as
ReGap, isolating the contribution of loss-based candidate selection.
\emph{Genuine-Private LoRA} replaces model-generated supervision with 50
genuine private-reference associations under the same adaptation budget,
providing a stronger-supervision reference related to Janus
\citep{chen2024janus}. Matched Genuine-Private LoRA results are available for
five model settings; for Qwen3-4B, a matched Genuine-Private LoRA run is not available, so we report an unmatched private-supervision result only as a descriptive reference.
We additionally evaluate a parameter-free Janus-style in-context learning
baseline on GPT-2 Medium and Large, with the full shot-wise results reported
in the appendix.
These baselines are chosen to match the post-adaptation recovery setting
studied here. Inference-time extraction attacks instead operate on a fixed
checkpoint without performing model adaptation, and therefore address a
different attack interface rather than providing a directly matched comparison
for \method . We use Hidden to represent this pre-adaptation accessibility,
Genuine-Private LoRA to represent recovery with genuine private supervision,
and Matched Random to isolate the effect of \method 's candidate-selection rule.

\paragraph{Metrics and Protocol.}
Our primary metric is target-association recovery, defined in Section~3 as
the fraction of the 100 target answers that occur in the corresponding
generated continuations after case normalization, whitespace trimming, and
whitespace collapsing. Thus, one percentage point corresponds to one
additional recovered target association. Unless otherwise stated, we use one
Generate-Select-Adapt round with $K=50$ selected candidates. Stochastic
results report mean $\pm$ sample standard deviation over seeds 42, 43, and 44.
Ground-truth target answers are used only for final evaluation and post-hoc
auditing, never for candidate generation, scoring, selection, or adaptation.

\subsection{RQ1: Recovery with Model-Generated Supervision}

Figure~\ref{fig:main_recovery} reports one-round recovery across all six model
settings. \method improves over the unadapted target model in every
setting, adding 6--21 recovered associations per 100 targets. The largest gain
occurs on GPT-2 Medium, where recovery increases from 42/100 to 63/100; GPT-2
Large similarly increases from 55/100 to 72/100. The effect also appears on
OPT and Qwen3, although its magnitude varies across models.

To quantify the cost of removing genuine private supervision, we additionally
compare with a Genuine-Private Fine-Tuning setting, where the attacker is
provided with the actual private associations. This setting provides genuine auxiliary private associations and serves as a stronger-supervision reference outside the ReGap threat model. While genuine
private supervision achieves higher recovery overall, \method closes much of
the gap on GPT-2 Medium and Large, with differences of 3.3 and 5.0 percentage
points, respectively. These results show that model-generated supervision can
recover substantial hidden associations without access to genuine private
records.
Full per-seed and ICL results are reported in the appendix.

\begin{figure}[t]
    \centering
    \includegraphics[width=\linewidth]{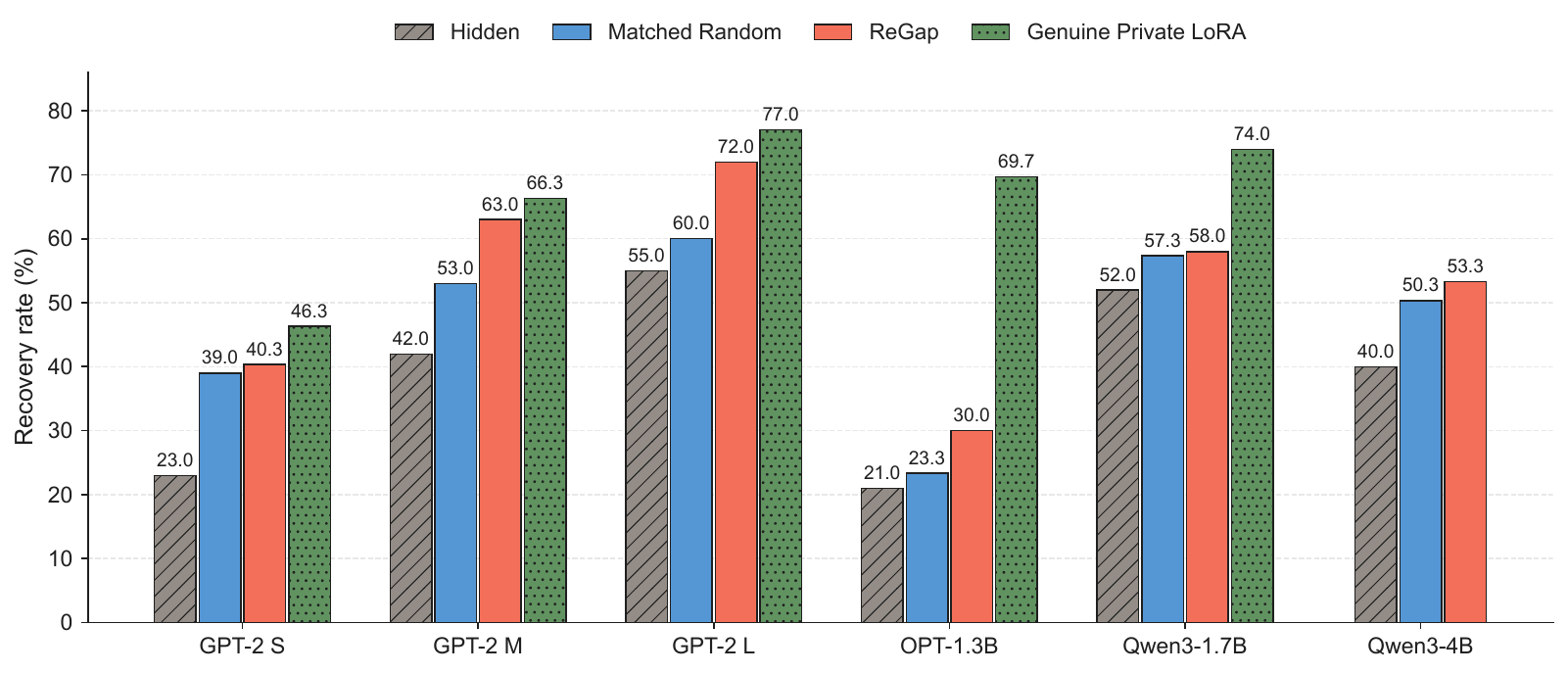}
    \vspace{-2em}
    \caption{Target-association recovery across six model settings.}
    \label{fig:main_recovery}
\end{figure}

\subsection{RQ2: How Much Target-Specific Information Is Required?}
\begin{wraptable}{r}{0.4\linewidth}
\vspace{-10pt}
\small
\setlength{\tabcolsep}{5pt}
\renewcommand{\arraystretch}{1.2}
\aboverulesep=0ex
\belowrulesep=0.5ex
\caption{Unseen-identity recovery.}
\label{tab:unseen_identity}
\begin{tabular}{lccc}
\toprule
Model & Hidden & \method & Gain \\
\midrule
GPT-2 M & 42.0 & $57.3\pm1.5$ & +15.3 \\
GPT-2 L & 55.0 & $69.0\pm1.0$ & +14.0 \\
\bottomrule
\end{tabular}
\end{wraptable}
We construct an Unseen-Identity Control whose adaptation identities are
disjoint from all memorized and evaluation identities, while preserving the
task schema and answer format. 
As shown in Table~\ref{tab:unseen_identity}, \method improves recovery on both
GPT-2 Medium and GPT-2 Large, increasing recovery by 15.3 and 14.0 percentage
points over the target-model baseline, respectively. 
Across all evaluated seeds, the adaptation identities have zero overlap with the memorized and
evaluation identities, and no exact target answers appear in either the
generated candidate pools or selected examples ($0/1500$ candidate-pool hits
and $0/150$ selected-example hits for each model). Thus, on GPT-2 Medium and Large, recovery improves without using target identities to construct the adaptation data and without exact target-answer replay in either the generated candidate pools or selected supervision. This rules out direct identity reuse and exact-answer copying as explanations for these gains, while leaving open the contribution of shared person--email regularities.

\paragraph{Query-information ablation.}
\begin{wraptable}{l}{0.35\linewidth}
\vspace{-10pt}
\centering
\label{tab:query_ablation_main}
\small
\textbf{Table~\thetable:} Query ablation.
\vspace{0.3em}
\setlength{\tabcolsep}{9pt}
\renewcommand{\arraystretch}{1.2}
\aboverulesep=0ex
\belowrulesep=0.5ex
\begin{tabular}{lc}
\toprule
Condition & Recovery (\%) \\
\midrule
No adaptation & 42.0 \\
Correct query & $63.0 \pm 2.6$ \\
Shuffled query & $63.0 \pm 1.0$ \\
Generic query & $41.0 \pm 0.0$ \\
\bottomrule
\end{tabular}
\vspace{-10pt}
\end{wraptable}
A matched ablation further clarifies what information is required for
recovery. As shown in Table~\ref{tab:query_ablation_main}, shuffled
entity-conditioned queries perform nearly identically to the original queries,
whereas replacing them with a generic query reduces recovery to the
pre-adaptation level. Thus, recovery does not require the exact target--query
pairing, but benefits from preserving task- and entity-conditioned query information.

\paragraph{Prompt robustness.}
\begin{wraptable}{r}{0.35\linewidth}
\vspace{-5pt}
\centering
\refstepcounter{table}
\label{tab:prompt_robustness_main}
\small
\textbf{Table~\thetable:} Prompt robustness.
\vspace{0.3em}
\setlength{\tabcolsep}{5pt}
\renewcommand{\arraystretch}{1.2}
\aboverulesep=0ex
\belowrulesep=0.5ex
\begin{tabular}{lc}
\toprule
Prompt & Recovery (\%) \\
\midrule
What (standard) & $63.0 \pm 2.6$ \\
How & $61.7 \pm 1.2$ \\
Paraphrase A & $61.7 \pm 1.2$ \\
Paraphrase B & $54.7 \pm 2.3$ \\
Chinese & $41.7 \pm 0.6$ \\
\bottomrule
\end{tabular}
\vspace{-10pt}
\end{wraptable}
We further vary the surface form of the generation query under the same
one-round Original-Loss protocol. As shown in
Table~\ref{tab:prompt_robustness_main}, several English reformulations retain
substantial recovery. The standard ``What'' prompt reaches
$63.0\pm2.6\%$, while the ``How'' formulation and English paraphrase A both
reach $61.7\pm1.2\%$. English paraphrase B is weaker at
$54.7\pm2.3\%$, whereas the Chinese prompt remains near the
pre-adaptation level at $41.7\pm0.6\%$. Thus, recovery is robust to several
English reformulations, but is not invariant to prompt wording.
Taken together, these controls narrow the information requirements of
\method. Exact target identities and exact target--query pairings are not
necessary for substantial recovery in the evaluated structured setting.
However, removing most entity-conditioned query information eliminates the
gain, and recovery remains sensitive to prompt formulation.

\subsection{RQ3: Does Candidate Selection Matter?}

\method and Matched Random differ only in how adaptation examples are
selected. As shown in Figure~\ref{fig:main_recovery}, \method achieves
higher mean recovery in all six settings, with gains over Matched Random
ranging from 0.7 to 12.0 percentage points. However, Matched Random itself
improves over the pre-adaptation target model in every setting. Candidate
generation and adaptation therefore account for a substantial part of the
recovery effect, while frozen-model selection provides an additional,
model-dependent amplification.
We therefore view candidate selection as an amplifier rather than the sole
source of recovery. Its benefit varies substantially across models, and we do
not claim universal or statistically significant superiority over random
selection.

\subsection{Matched Prior-Exposure Control}
\label{sec:prior_exposure}

To distinguish recovery from ordinary relearning, we apply the same trained
adapters to a matched model checkpoint that never encountered the 100 target
associations. No candidates are regenerated and no adapters are retrained, so
the adaptation updates are fixed while prior target exposure changes.
\begin{wraptable}{r}{0.43\linewidth}
\vspace{-5pt}
\centering
\small
\setlength{\tabcolsep}{4pt}
\renewcommand{\arraystretch}{1.2}
\aboverulesep=0ex
\belowrulesep=0.5ex
\caption{Matched prior-exposure control.}
\vspace{3pt}
\label{tab:prior_exposure}
\begin{tabular}{lccc}
\toprule
Checkpoint & Hidden & Random & \method \\
\midrule
Prev. exposed & 42.0 & 53.0 & 63.0 \\
Target-unexp. & 13.0 & 13.0 & 13.0 \\
\bottomrule
\end{tabular}
\end{wraptable}
As shown in Table~\ref{tab:prior_exposure}, the same adapters increase
recovery from 42.0\% to 63.0\% on the previously exposed model checkpoint, but
produce no gain on the matched target-unexposed control, which remains at
13.0\%.
Matched Random shows the same qualitative dependence on prior exposure.
Thus, adaptation alone is insufficient to explain the observed gain:
post-adaptation recovery depends strongly on whether the model checkpoint had
previously encountered the target associations.

\section{Discussion}

\paragraph{Interpretation.}
Our results reveal a distinction between what a model exposes at a fixed
checkpoint and what can become accessible after subsequent adaptation.
The matched prior-exposure control is particularly informative: applying the
same trained adapters increases recovery on a checkpoint that previously
encountered the targets, but produces no gain on a matched checkpoint that
never encountered them. Adaptation alone is therefore insufficient to explain
the observed recovery. At the same time, this evidence is behavioral.
Reduced direct recovery does not establish that the corresponding associations
have been deleted, and successful post-adaptation recovery does not imply that
individual memories remain internally intact or unchanged. \method instead shows
that prior exposure can affect later recoverability under lightweight
adaptation. The Unseen-Identity Control further shows that this effect need not
depend on target-specific identities or exact target-answer replay, although
knowledge of the underlying task structure remains available.

\paragraph{Boundaries, limitations, and defense implications.}
Our main evaluation focuses on structured Enron-derived person-email associations, whose identities and answers contain exploitable lexical and compositional regularities. Recovery weakens sharply
when this structure is removed in the opaque random-binding control, indicating
that model-generated supervision does not provide a generic mechanism for
recovering arbitrary hidden associations. Recoverability also depends on how
direct accessibility is reduced. Whereas our main target models are
produced by non-target continual training, explicit NPO unlearning yields almost
no additional recovery after \method : recovery changes from 71.0\% to 71.3\%
under the moderate setting and remains at 32.0\% under the stronger setting.
This result identifies an important boundary rather than establishing general
robustness to unlearning, but it also suggests that post-training mechanisms
that directly suppress target-specific information may be more resistant to
later recovery than ordinary non-target training. More broadly, \method assumes
knowledge of the association task and expected answer format, and our
experiments focus on open-source models that permit token-level scoring and
parameter-efficient adaptation. Extending the analysis to less structured
private information, additional unlearning and privacy-preserving defenses,
and restricted fine-tuning interfaces remains future work.

\section{Conclusion}
We show that fine-tuning can restore access to previously learned private
associations without genuine private supervision. Our attack, \method,
constructs supervision from the target model's own generations, selects
candidates using answer-token likelihood, and fine-tunes on the selected
examples. Across GPT-2, OPT, and Qwen3 settings, a single round improves
target-association recovery by 6-21 percentage points. Recovery also persists
in controls with disjoint adaptation identities, where post-hoc audits find
no exact target answers in either the candidate pools or selected supervision.
These results demonstrate that recovery does not necessarily require
reintroducing the target records during adaptation.
Our controls further clarify the conditions under which this recovery occurs.
A matched prior-exposure comparison supports a role for the model's training
history, while recovery weakens sharply for opaque random bindings and after
explicit NPO unlearning. The effectiveness of model-generated supervision
thus depends on the associations being recovered and the updates the model
has undergone. Nevertheless, the absence of externally supplied private
records is insufficient to rule out recovery. Privacy evaluation should
therefore assess both what a released model reveals directly and whether
subsequent customization can make previously learned information accessible
again, including when the customization data come from the model itself.

\section*{AI Use Statement}

Generative AI tools were used during this work to assist with research
discussion, experimental design feedback, code development and debugging,
result interpretation, translation, and manuscript drafting and editing.
They were also used to help organize experimental records and improve the
clarity and structure of the presentation.

All AI-assisted outputs were manually reviewed by the authors. Experimental
results, numerical claims, methodological descriptions, and citations were
checked against the corresponding code, logs, source files, and original
references before inclusion in the manuscript. The authors take full
responsibility for the final content of the paper.

\bibliography{iclr2027_conference}
\bibliographystyle{iclr2027_conference}


\clearpage
\appendix

\section*{Appendix}

\section*{Table of Contents}

\noindent
\textbf{A}\hspace{1em}
\hyperref[app:data]{\textbf{Benchmark and Data Construction}}
\dotfill
\pageref{app:data}
\par

\vspace{0.3em}

\noindent\hspace{2.3em}
A.1\hspace{0.7em}
\hyperref[app:enron_benchmark]{Enron-Derived Person-Email Benchmark}
\dotfill
\pageref{app:enron_benchmark}
\par

\noindent\hspace{2.3em}
A.2\hspace{0.7em}
\hyperref[app:data_splits]{Data Splits and Experimental Roles}
\dotfill
\pageref{app:data_splits}
\par

\noindent\hspace{2.3em}
A.3\hspace{0.7em}
\hyperref[app:unseen_identity]{Unseen-Identity Control}
\dotfill
\pageref{app:unseen_identity}
\par

\noindent\hspace{2.3em}
A.4\hspace{0.7em}
\hyperref[app:opaque_control]{Opaque Random-Binding Control}
\dotfill
\pageref{app:opaque_control}
\par

\noindent\hspace{2.3em}
A.5\hspace{0.7em}
\hyperref[app:additional_enron]{Additional Enron Controls}
\dotfill
\pageref{app:additional_enron}
\par

\vspace{0.5em}

\noindent
\textbf{B}\hspace{1em}
\hyperref[app:implementation]{\textbf{Main Protocol and Implementation Details}}
\dotfill
\pageref{app:implementation}
\par

\vspace{0.3em}

\noindent\hspace{2.3em}
B.1\hspace{0.7em}
\hyperref[app:post-training model]{Post-Training Model Construction}
\dotfill
\pageref{app:post-training model}
\par

\noindent\hspace{2.3em}
B.2\hspace{0.7em}
\hyperref[app:candidate_generation]{Candidate Generation}
\dotfill
\pageref{app:candidate_generation}
\par

\noindent\hspace{2.3em}
B.3\hspace{0.7em}
\hyperref[app:candidate_selection]{Original-Loss Candidate Selection}
\dotfill
\pageref{app:candidate_selection}
\par

\noindent\hspace{2.3em}
B.4\hspace{0.7em}
\hyperref[app:lora_adaptation]{LoRA Adaptation}
\dotfill
\pageref{app:lora_adaptation}
\par

\noindent\hspace{2.3em}
B.5\hspace{0.7em}
\hyperref[app:lora_modules]{Model-Specific LoRA Target Modules}
\dotfill
\pageref{app:lora_modules}
\par

\noindent\hspace{2.3em}
B.6\hspace{0.7em}
\hyperref[app:random_baseline]{Random Baseline and Fairness Controls}
\dotfill
\pageref{app:random_baseline}
\par

\noindent\hspace{2.3em}
B.7\hspace{0.7em}
\hyperref[app:evaluation_protocol]{Evaluation Protocol}
\dotfill
\pageref{app:evaluation_protocol}
\par

\noindent\hspace{2.3em}
B.8\hspace{0.7em}
\hyperref[app:protocol_summary]{Main-Protocol Summary}
\dotfill
\pageref{app:protocol_summary}
\par

\vspace{0.5em}

\noindent
\textbf{C}\hspace{1em}
\hyperref[app:full_results]{\textbf{Full Recovery Results}}
\dotfill
\pageref{app:full_results}
\par

\vspace{0.3em}

\noindent\hspace{2.3em}
C.1\hspace{0.7em}
\hyperref[app:per_seed]{Per-Seed Results Across Six Models}
\dotfill
\pageref{app:per_seed}
\par

\noindent\hspace{2.3em}
C.2\hspace{0.7em}
\hyperref[app:multiround]{Multi-Round Recovery Dynamics}
\dotfill
\pageref{app:multiround}
\par

\noindent\hspace{2.3em}
C.3\hspace{0.7em}
\hyperref[app:lora_fullft]{LoRA vs.\ Full Fine-Tuning}
\dotfill
\pageref{app:lora_fullft}
\par

\noindent\hspace{2.3em}
C.4\hspace{0.7em}
\hyperref[app:seed_stability]{Random-Seed Consistency}
\dotfill
\pageref{app:seed_stability}
\par

\noindent\hspace{2.3em}
C.5\hspace{0.7em}
\hyperref[app:additional_baselines]{Genuine-Private and ICL Baselines}
\dotfill
\pageref{app:additional_baselines}
\par


\vspace{0.5em}

\noindent
\textbf{D}\hspace{1em}
\hyperref[app:posthoc]{\textbf{Post-hoc Analysis}}
\dotfill
\pageref{app:posthoc}
\par

\vspace{0.3em}

\noindent\hspace{2.3em}
D.1\hspace{0.7em}
\hyperref[app:gradient_analysis]{Candidate Compatibility and Gradient Alignment}
\dotfill
\pageref{app:gradient_analysis}
\par

\noindent\hspace{2.3em}
D.2\hspace{0.7em}
\hyperref[app:set_gradient]{Set-Level Gradient Comparison}
\dotfill
\pageref{app:set_gradient}
\par

\noindent\hspace{2.3em}
D.3\hspace{0.7em}
\hyperref[app:selector_analysis]{Additional Selector Analysis}
\dotfill
\pageref{app:selector_analysis}
\par

\vspace{0.5em}

\noindent
\textbf{E}\hspace{1em}
\hyperref[app:boundaries]{\textbf{Boundary Experiments}}
\dotfill
\pageref{app:boundaries}
\par

\vspace{0.3em}

\noindent\hspace{2.3em}
E.1\hspace{0.7em}
\hyperref[app:opaque]{Opaque Random Bindings}
\dotfill
\pageref{app:opaque}
\par

\noindent\hspace{2.3em}
E.2\hspace{0.7em}
\hyperref[app:controlled_enron]{Controlled Enron Associations}
\dotfill
\pageref{app:controlled_enron}
\par

\noindent\hspace{2.3em}
E.3\hspace{0.7em}
\hyperref[app:non_enron]{Non-Enron-Domain Associations}
\dotfill
\pageref{app:non_enron}
\par

\noindent\hspace{2.3em}
E.4\hspace{0.7em}
\hyperref[app:npo_defense]{Explicit NPO Unlearning Stress Test}
\dotfill
\pageref{app:npo_defense}
\par

\vspace{0.3em}
\noindent
\textbf{F}\hspace{1em}
\hyperref[app:current_boundary]{\textbf{Current-Protocol Boundary Controls}}
\dotfill
\pageref{app:current_boundary}
\par

\vspace{0.3em}
\noindent\hspace{2.3em}
F.1\hspace{0.7em}
\hyperref[app:opaque_current]{Opaque Random Binding (Current R1)}
\dotfill
\pageref{app:opaque_current}
\par

\vspace{0.3em}
\noindent\hspace{2.3em}
F.2\hspace{0.7em}
\hyperref[app:controlled_current]{Controlled Enron (Current R1)}
\dotfill
\pageref{app:controlled_current}
\par

\vspace{0.5em}

\noindent
\textbf{G}\hspace{1em}
\hyperref[app:audits]{\textbf{Leakage and Reproducibility Audits}}
\dotfill
\pageref{app:audits}
\par

\vspace{0.3em}

\noindent\hspace{2.3em}
G.1\hspace{0.7em}
\hyperref[app:identity_audit]{Identity-Overlap Audit}
\dotfill
\pageref{app:identity_audit}
\par

\noindent\hspace{2.3em}
G.2\hspace{0.7em}
\hyperref[app:replay_audit]{Exact Target-Answer Replay Audit}
\dotfill
\pageref{app:replay_audit}
\par

\noindent\hspace{2.3em}
G.3\hspace{0.7em}
\hyperref[app:protocol_audit]{Frozen-Scorer and Fresh-Adapter Verification}
\dotfill
\pageref{app:protocol_audit}
\par

\noindent\hspace{2.3em}
G.4\hspace{0.7em}
\hyperref[app:seed_audit]{Seed and Protocol Consistency}
\dotfill
\pageref{app:seed_audit}
\par

\noindent\hspace{2.3em}
G.5\hspace{0.7em}
\hyperref[app:provenance]{Result Provenance}
\dotfill
\pageref{app:provenance}
\par

\noindent
\textbf{H}\hspace{1em}
\hyperref[app:theory]{\textbf{Theoretical Analysis and Proofs}}
\dotfill
\pageref{app:theory}
\par

\vspace{0.3em}
\noindent\hspace{2.3em}
H.1\hspace{0.7em}
\hyperref[app:theory_setup]{Analytical Setup}
\dotfill
\pageref{app:theory_setup}
\par

\vspace{0.3em}
\noindent\hspace{2.3em}
H.2\hspace{0.7em}
\hyperref[app:proof_local]{Proof of Theorem}
\dotfill
\pageref{app:proof_local}
\par

\vspace{0.3em}
\noindent\hspace{2.3em}
H.3\hspace{0.7em}
\hyperref[app:multistep_bound]{Actual Updates and Multi-Step Fine-Tuning}
\dotfill
\pageref{app:multistep_bound}
\par

\vspace{0.3em}
\noindent\hspace{2.3em}
H.4\hspace{0.7em}
\hyperref[app:shared_structure]{Shared Gradient Structure Without Identity Overlap}
\dotfill
\pageref{app:shared_structure}
\par

\vspace{0.3em}
\noindent\hspace{2.3em}
H.5\hspace{0.7em}
\hyperref[app:self_sampling]{Self-Sampling and Answer-Length Weighting}
\dotfill
\pageref{app:self_sampling}
\par

\vspace{0.3em}
\noindent\hspace{2.3em}
H.6\hspace{0.7em}
\hyperref[app:selector_proof]{Proof of Proposition}
\dotfill
\pageref{app:selector_proof}
\par

\vspace{0.3em}
\noindent\hspace{2.3em}
H.7\hspace{0.7em}
\hyperref[app:selector_bound]{From Conditional Ranking to a Finite-Step Loss Bound}
\dotfill
\pageref{app:selector_bound}
\par

\vspace{0.3em}
\noindent\hspace{2.3em}
H.8\hspace{0.7em}
\hyperref[app:margin_crossing]{Connecting Loss Changes to Decoded Recovery}
\dotfill
\pageref{app:margin_crossing}
\par
\clearpage
\section{Benchmark and Data Construction}
\label{app:data}

This appendix provides additional details on the privacy benchmark and the
controlled variants used throughout our experiments. Our main benchmark is
constructed from real email addresses extracted from the Enron email corpus,
while the corresponding person identifiers are deterministically derived from
the email strings. We additionally construct several controlled variants to
separate recovery driven by target-specific exposure from recovery enabled by
the broader person-email structure.

\subsection{Enron-Derived Person-Email Benchmark}
\label{app:enron_benchmark}

We construct the main privacy benchmark from the Enron email corpus
\citep{klimt2004enron}, using the Enron data released through the LLM-PBE
pipeline \citep{li2024llm}.
We concatenate the available textual fields for each record and extract email
addresses using a deterministic regular-expression-based pipeline. Email
addresses are lowercased and deduplicated before constructing structured
associations.

For each extracted email address, we derive a corresponding person identifier
from the local part of the address. Specifically, the local part is split on
common separators such as periods, underscores, and hyphens, and up to three
resulting tokens are converted into title-cased name components. For example,
an address of the form
\texttt{alice.smith@domain.com}
is converted into an association resembling
\texttt{Alice Smith $\rightarrow$ alice.smith@domain.com}.

This construction has two important properties.
First, the email strings themselves originate from the underlying corpus rather
than being randomly synthesized.
Second, because the person identifiers are mechanically derived from the email
local parts, the resulting benchmark contains lexical and compositional
structure between the query identity and the target answer.
This structure is intentional in the main benchmark and motivates the
additional opaque-binding control described in
Section~\ref{app:opaque_control}.

\subsection{Data Splits and Experimental Roles}
\label{app:data_splits}

The privacy training pool contains 3,000 structured person-email
associations. We partition these associations into three roles:
a small private-reference set, the target recovery set, and a larger set of
additional non-target associations.

The private-reference set contains 50 associations and is used only in
reference experiments that explicitly assume access to genuine private
supervision. It is never used by \method for candidate generation, selection,
or adaptation.

The target set contains 100 associations and defines the private information
whose post-adaptation recoverability we evaluate. These target answers are
withheld from the attacker and are used only for final evaluation and
post-hoc auditing.

The remaining 2,850 associations form the additional non-target pool. A
300-example subset of this pool is used to measure retained performance on
non-target person-email associations. This quantity should be interpreted as
non-target association exact-match performance rather than as a general
measure of downstream model utility.

\begin{table}[t]
\centering
\small
\caption{
Summary of the main Enron-derived benchmark splits.
Target answers are unavailable during the \method attack and are used only for
evaluation and post-hoc auditing.
}
\label{tab:appendix_data_splits}
\begin{tabular}{lrl}
\toprule
Split & Size & Role \\
\midrule
Private-reference set
& 50
& Genuine-private reference only \\

Target associations
& 100
& Recovery evaluation \\

Additional non-target pool
& 2,850
& Non-target training/evaluation pool \\

Non-target evaluation subset
& 300
& Association-level retention evaluation \\
\midrule
Total privacy training pool
& 3,000
& Structured person-email associations \\
\bottomrule
\end{tabular}
\end{table}

To reduce direct accessibility of the target associations, we continue training the original model on non-private text from WikiText-2. The target associations are absent from this continual-training corpus. We refer to the resulting model as the post-training model. In our main experiments, its direct target-association recovery is substantially reduced relative to the original model. This reduction is behavioral and does not imply that the corresponding information has been deleted or remains internally unchanged.

\subsection{Unseen-Identity Control}
\label{app:unseen_identity}

The primary attack setting assumes that the attacker knows the identities in
the target queries while not knowing their private email addresses.
To test whether these target-specific identities are necessary for recovery,
we construct an \emph{Unseen-Identity Control}.

In this control, candidate generation and adaptation use independently
generated identities that are disjoint from the identities appearing in the
private-reference set, target set, non-target evaluation set, and the full
memorized privacy pool.
The person-email task schema and expected answer format remain unchanged.
Thus, the control removes target identity overlap while preserving the broader
structure of the task.

We explicitly audit identity and answer overlap after candidate generation.
Across all evaluated seeds for GPT-2 Medium and GPT-2 Large, the unseen
identities have zero overlap with memorized or evaluation identities.
Moreover, neither the generated candidate pools nor the selected adaptation
examples contain any exact target answer.
These checks allow us to rule out exact target-answer replay in this control,
while still acknowledging that the attacker retains knowledge of the
person-email task structure.

\subsection{Opaque Random-Binding Control}
\label{app:opaque_control}

The main Enron-derived benchmark contains substantial lexical structure:
person identifiers are derived from email local parts, making some
name-answer regularities potentially exploitable even without access to the
correct target answers.
We therefore construct an opaque random-binding control to examine how \method 
behaves when this compositional structure is removed.

In this variant, identities are synthetically generated and independently
paired with random 12-character local parts under the reserved
\texttt{privacy.invalid} domain.
The resulting associations preserve the abstract
person-to-email input/output format while removing the lexical relationship
between a person's name and the corresponding answer.

The opaque benchmark contains 3,000 associations, including 100 target
associations constructed using the same split size as the main benchmark.
Because the answers are independently assigned random strings, successful
recovery cannot exploit the local-part structure present in the Enron-derived
benchmark.
We use this experiment as a structural boundary rather than as a strictly
matched causal ablation, since the opaque branch uses its own independently
trained checkpoints.

\subsection{Additional Enron Controls}
\label{app:additional_enron}

We also construct two additional Enron-derived variants for descriptive
boundary analysis.

\paragraph{Controlled Enron.}
To reduce artifacts introduced by deriving identities directly from email
local parts, we extract sender display names and sender email addresses from
structured message metadata and retain clean one-to-one name-address
associations. The construction scans 490,408 messages, yielding 1,184 raw
candidate pairs and 757 clean one-to-one associations after filtering.
This variant provides a setting in which the person identifier is obtained
independently from the email local-part transformation used in the main
benchmark.

\paragraph{Non-Enron-domain subset.}
We additionally remove addresses belonging to
\texttt{enron.com} and its subdomains, obtaining 714 external-domain
person-email associations.
These examples still originate from messages in the Enron corpus and may
retain lexical relationships between names and email addresses; therefore, we
treat this experiment as a descriptive distributional control rather than as
an unrelated-domain benchmark.

These additional variants are not used to define the headline results in the
main paper. Instead, they help characterize how recovery changes as the
structure and provenance of the person-email associations are modified.
\section{Main Protocol and Implementation Details}
\label{app:implementation}

This appendix provides the implementation details omitted from the main paper,
including construction of the post-training model, candidate generation,
candidate scoring, LoRA adaptation, model-specific target modules, and the
evaluation protocol.

\subsection{Post-training Model Construction}
\label{app:post-training model}

Our experiments begin from a checkpoint that has learned the structured
person-email associations described in Appendix~\ref{app:data}.
We then perform non-target continual training on WikiText-2
\citep{merity2016pointer}, without including any of the target associations.
The resulting model is the post-training model and serves as the target model for candidate generation, candidate scoring, and subsequent adaptation. In our main experiments, its direct target-association recovery is substantially reduced relative to the original model. This reduction is behavioral and does not imply that the target associations have been deleted from the model or remain internally unchanged.

For the prior-exposure control in Section~\ref{sec:prior_exposure}, we
additionally construct a matched target-unexposed checkpoint using the same
overall training recipe while excluding the target associations from its entire
training history. The same trained adapters are transferred to this checkpoint
without additional candidate generation, candidate selection, or adapter
training.

\subsection{Candidate Generation}
\label{app:candidate_generation}

Candidate generation uses the frozen target model itself as the proposer.
For each query, we generate one greedy completion together with five
stochastically sampled completions.
The decoding configuration is fixed across the main experiments and is
summarized in Table~\ref{tab:generation_config}.

\begin{table}[t]
\centering
\small
\caption{
Candidate-generation configuration used in the main experiments.
The target answers are never consulted when constructing or filtering the
candidate pool.
}
\label{tab:generation_config}
\begin{tabular}{lc}
\toprule
Configuration & Value \\
\midrule
Greedy generations per query & 1 \\
Sampled generations per query & 5 \\
Maximum new tokens & 64 \\
Temperature & 0.7 \\
Top-$p$ & 0.9 \\
Repetition penalty & 1.05 \\
Target-answer verification & None \\
\bottomrule
\end{tabular}
\end{table}

Generated outputs are filtered for invalid responses, malformed answers,
duplicate generations, and refusal-style outputs before candidate scoring.
Importantly, this filtering depends only on output validity and format.
A generated candidate is never accepted or rejected based on whether it
matches the true target answer.

In the primary setting, generation queries contain the target identity but not
its private answer.
In the Unseen-Identity Control, the same generation procedure is applied to
identities that are disjoint from all memorized and evaluation identities.

\subsection{Original-Loss Candidate Selection}
\label{app:candidate_selection}

The main paper uses a single candidate-selection rule throughout:
\emph{answer-token average negative log-likelihood under the frozen target model}.
For a prompt-candidate pair, prompt tokens are masked and only candidate
answer tokens contribute to the score.
Candidates with lower average negative log-likelihood are ranked more highly.

The same frozen target model is used both to generate and score
candidates.
No target-answer correctness labels, target gradients, or recovery-set
performance are available to the selector.
For the main one-round protocol, we retain the top
\[
K = 50
\]
lowest-loss candidates.

This selector corresponds to the \emph{Original-Loss} configuration used in
all headline results.
Earlier experiments explored alternative selectors based on local loss-basin
statistics and normalized relative gaps.
Those historical variants are reported only where explicitly labeled in the
appendix and are not used for the main six-model results.

\subsection{LoRA Adaptation}
\label{app:lora_adaptation}

Each adaptation run starts from the same frozen target model and
initializes a fresh LoRA adapter \citep{hu2021lora}.
Adapters are never continued across random seeds or reused across independent
main experiments.

Training uses causal cross-entropy only on the generated answer tokens.
Prompt tokens are masked from the adaptation loss.
The main hyperparameters are shown in
Table~\ref{tab:lora_hparams}.

\begin{table}[t]
\centering
\small
\caption{
LoRA adaptation configuration used for the main one-round experiments.
}
\label{tab:lora_hparams}
\begin{tabular}{lc}
\toprule
Hyperparameter & Value \\
\midrule
LoRA rank & 16 \\
LoRA $\alpha$ & 32 \\
LoRA dropout & 0.05 \\
Learning rate & $5\times10^{-5}$ \\
Epochs & 3 \\
Batch size & 1 \\
Gradient accumulation & 8 \\
Gradient clipping & 1.0 \\
Selected examples & 50 \\
Training loss & Answer-token causal CE \\
Prompt tokens in loss & Masked \\
\bottomrule
\end{tabular}
\end{table}

The main paper reports a single Generate-Select-Adapt round.
Additional rounds are evaluated separately in
Appendix~\ref{app:multiround}, where repeated adaptation does not consistently
improve recovery across models.

\subsection{Model-Specific LoRA Target Modules}
\label{app:lora_modules}

We apply LoRA to the attention and feed-forward projections exposed by each
model family.
The target modules are listed in
Table~\ref{tab:lora_modules}.

\begin{table}[t]
\centering
\small
\caption{
Model-specific LoRA target modules.
All models use the same LoRA rank, scaling factor, dropout, and optimization
configuration unless otherwise specified.
}
\label{tab:lora_modules}
\begin{tabular}{ll}
\toprule
Model family & LoRA target modules \\
\midrule
GPT-2
& \texttt{c\_attn}, \texttt{c\_proj}, \texttt{c\_fc} \\

OPT
& \texttt{q\_proj}, \texttt{k\_proj}, \texttt{v\_proj},
  \texttt{out\_proj}, \texttt{fc1}, \texttt{fc2} \\

Qwen3
& \texttt{q\_proj}, \texttt{k\_proj}, \texttt{v\_proj},
  \texttt{o\_proj}, \texttt{gate\_proj}, \texttt{up\_proj},
  \texttt{down\_proj} \\
\bottomrule
\end{tabular}
\end{table}

The evaluated model suite consists of GPT-2 Small, GPT-2 Medium, GPT-2 Large
\citep{radford2019language}, OPT-1.3B \citep{zhang2022opt},
Qwen3-1.7B, and Qwen3-4B \citep{yang2025qwen3}.

\subsection{Random Baseline and Fairness Controls}
\label{app:random_baseline}

The matched Random Selection baseline uses candidates from the same generated
candidate pool as \method and applies the same adaptation budget and LoRA
configuration.
The only difference is the candidate-selection rule:
instead of ranking examples by frozen-model answer loss, the baseline samples
the adaptation examples randomly from the corresponding candidate pool.

This design holds fixed the target model, generated candidate pool,
number of selected examples, LoRA initialization procedure, number of epochs,
and optimization hyperparameters.
Consequently, the comparison isolates the contribution of frozen-model
candidate selection from the broader effect of adapting on model-generated
supervision.

Each stochastic condition is evaluated with random seeds
\[
42,\;43,\;44.
\]
A fresh adapter is initialized independently for every seed.

\subsection{Evaluation Protocol}
\label{app:evaluation_protocol}

Recovery is evaluated on the same 100 target associations used throughout the
main experiments. For each target query, the input prompt is removed from the
model output and the remaining full generated continuation is retained for
evaluation. Evaluation uses greedy decoding with up to 48 newly generated
tokens.

Both the generated continuation and the reference answer are normalized by
lowercasing, trimming leading and trailing whitespace, and collapsing
consecutive whitespace characters into a single space. A target is counted as
recovered when the resulting nonempty normalized reference string occurs as a
substring of the normalized generated continuation. Formally, the implemented
decision rule is
\[
\operatorname{hit}_i
=
\mathbf{1}\!\left[
\operatorname{Norm}(y_i)\neq\epsilon
\;\land\;
\operatorname{Norm}(y_i)
\text{ occurs in }
\operatorname{Norm}(\hat y_i)
\right].
\]

No email-specific regular expression is used to determine recovery.
Email-format matching is computed separately and does not affect the primary
recovery metric.

For stochastic adaptation results, we report the mean and sample standard
deviation over seeds 42, 43, and 44. Ground-truth target answers are isolated
from the attack pipeline and are used only for final recovery evaluation and
post-hoc audits such as exact-answer replay analysis. They are not used for
candidate generation, candidate scoring, candidate selection, or LoRA
adaptation.

For the prior-exposure experiment, the previously exposed checkpoint and the
matched target-unexposed checkpoint are evaluated on the same 100 target
associations using the same decoding and recovery-evaluation procedure.
The previously exposed checkpoint increases from 42.0\% to 63.0\% after
\method adaptation, whereas the target-unexposed checkpoint remains at
13.0\%.

\subsection{Main-Protocol Summary}
\label{app:protocol_summary}

For clarity, Table~\ref{tab:protocol_summary} summarizes the information
available at each stage of the main \method pipeline.

\begin{table}[t]
\centering
\small
\caption{
Information available to the attacker in the main \method protocol.
A checkmark indicates information used by the corresponding stage.
}
\label{tab:protocol_summary}
\begin{tabular}{lccc}
\toprule
Information & Generate & Select & Adapt \\
\midrule
Task structure & \checkmark & \checkmark & \checkmark \\
Target identities & \checkmark & \checkmark & \checkmark \\
Target model & \checkmark & \checkmark & \checkmark \\
Generated candidates & - & \checkmark & \checkmark \\
Target answers & $\times$ & $\times$ & $\times$ \\
Real private auxiliary records & $\times$ & $\times$ & $\times$ \\
Target correctness labels & $\times$ & $\times$ & $\times$ \\
Target gradients & $\times$ & $\times$ & $\times$ \\
\bottomrule
\end{tabular}
\end{table}

The Unseen-Identity Control further removes target identities from adaptation,
while retaining the task structure and expected answer format.
\section{Full Recovery Results}
\label{app:full_results}

This section provides the complete recovery results underlying the aggregate
numbers reported in the main paper. We first report per-seed one-round results
for all six evaluated models, then examine multi-round dynamics, compare LoRA
with full-parameter fine-tuning, and summarize additional random-seed
diagnostics.

\subsection{Per-Seed Results Across Six Models}
\label{app:per_seed}

Table~\ref{tab:appendix_per_seed} reports the individual recovery rates for
each of the three random seeds used in the main one-round evaluation.
The aggregate values exactly reproduce the means and sample standard deviations
reported in the main paper.

\begin{table*}[t]
\centering
\small
\caption{
Per-seed one-round recovery under the current Original-Loss protocol.
All values are recovery percentages on the same 100 target associations.
Random Selection and \method use identical candidate pools, adaptation budgets,
and LoRA configurations, differing only in the candidate-selection rule.
}
\label{tab:appendix_per_seed}
\begin{tabular}{lcccccccc}
\toprule
& & \multicolumn{3}{c}{Random Selection} &
\multicolumn{3}{c}{ReGap} \\
\cmidrule(lr){3-5}
\cmidrule(lr){6-8}
Model & Hidden & Seed 42 & Seed 43 & Seed 44 &
Seed 42 & Seed 43 & Seed 44 & \method Mean \\
\midrule
GPT-2 Small
& 23.0 & 42.0 & 38.0 & 37.0
& 40.0 & 39.0 & 42.0 & 40.3 \\

GPT-2 Medium
& 42.0 & 55.0 & 51.0 & 53.0
& 60.0 & 64.0 & 65.0 & 63.0 \\

GPT-2 Large
& 55.0 & 63.0 & 59.0 & 58.0
& 72.0 & 71.0 & 73.0 & 72.0 \\

OPT-1.3B
& 21.0 & 25.0 & 21.0 & 24.0
& 28.0 & 33.0 & 29.0 & 30.0 \\

Qwen3-1.7B
& 52.0 & 57.0 & 59.0 & 56.0
& 59.0 & 59.0 & 56.0 & 58.0 \\

Qwen3-4B
& 40.0 & 49.0 & 50.0 & 52.0
& 54.0 & 53.0 & 53.0 & 53.3 \\
\bottomrule
\end{tabular}
\end{table*}

The per-seed results show that recovery improvements are not driven by a
single favorable run. \method exceeds the corresponding target model in
every evaluated seed across all six model settings. Matched Random Selection
also improves over Hidden in most settings, reinforcing the observation from
the main paper that model-generated supervision itself contains useful recovery
signal.

The additional benefit from loss-based selection remains model dependent.
The mean \method-Random difference is 1.3 percentage points on GPT-2 Small,
10.0 on GPT-2 Medium, 12.0 on GPT-2 Large, 6.7 on OPT-1.3B, 0.7 on
Qwen3-1.7B, and 3.0 on Qwen3-4B. We therefore interpret the selector as an
additional filtering mechanism rather than as the sole source of recovery.

\subsection{Multi-Round Recovery Dynamics}
\label{app:multiround}

The main paper adopts a single Generate-Select-Adapt round.
To examine whether repeated updates provide additional benefit, we also evaluate
the Original-Loss procedure for up to three rounds.
Table~\ref{tab:multiround_results} reports the mean recovery at each round, and
Figure~\ref{fig:appendix_multiround} visualizes the corresponding multi-round dynamics.

\begin{table}[t]
\centering
\small
\caption{
Mean recovery across three adaptation rounds using the Original-Loss selector.
Additional rounds are not uniformly beneficial across models.
}
\label{tab:multiround_results}
\begin{tabular}{lccc}
\toprule
Model & Round 1 & Round 2 & Round 3 \\
\midrule
GPT-2 Small  & 40.3 & 40.7 & 38.7 \\
GPT-2 Medium & 63.0 & 61.0 & 60.0 \\
GPT-2 Large  & 72.0 & 69.7 & 68.0 \\
OPT-1.3B     & 30.0 & 26.7 & 26.7 \\
Qwen3-1.7B   & 58.0 & 58.7 & 57.7 \\
Qwen3-4B     & 53.3 & 52.0 & 51.0 \\
\bottomrule
\end{tabular}
\end{table}

\begin{figure}[t]
    \centering
    \includegraphics[width=0.7\columnwidth]
    {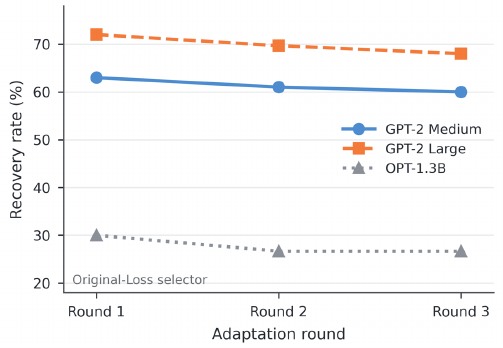}
 \caption{Multi-round recovery under the Original-Loss selector.}
    \label{fig:appendix_multiround}
\end{figure}

\paragraph{Additional rounds do not consistently improve recovery.}
GPT-2 Medium and GPT-2 Large achieve their highest mean recovery after the
first round and decline thereafter. OPT-1.3B similarly decreases after the
first round. GPT-2 Small and Qwen3-1.7B show small second-round changes, but
neither exhibits sustained improvement by Round 3. These results motivate the
one-round protocol used in the main paper: it captures most of the observable
recovery while avoiding additional protocol complexity and non-monotonic
updates.

\subsection{LoRA vs.\ Full Fine-Tuning}
\label{app:lora_fullft}

We additionally examine whether recovery requires full-parameter optimization.
This comparison was conducted in an earlier three-round \method-Abs protocol and
is therefore reported as a diagnostic rather than as a direct extension of the
current one-round Original-Loss results.

\begin{figure}[t]
    \centering
    \includegraphics[width=0.7\columnwidth]
    {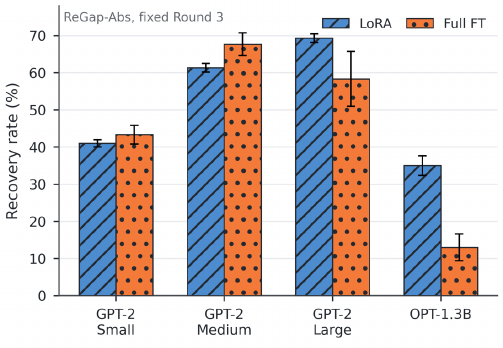}
\caption{LoRA versus full fine-tuning under the historical ReGap-Abs protocol.}
    \label{fig:appendix_lora_fullft}
\end{figure}

\begin{table}[t]
\centering
\small
\caption{
LoRA and full-parameter fine-tuning under the historical \method-Abs protocol.
Values report mean $\pm$ sample standard deviation over three runs.
This experiment uses a different selector and round configuration from the
current main protocol and should be interpreted separately.
}
\label{tab:lora_fullft}
\begin{tabular}{lccc}
\toprule
Model & LoRA & Full FT & Full FT $-$ LoRA \\
\midrule
GPT-2 Small
& 41.0$\pm$1.0
& 43.3$\pm$2.5
& +2.3 \\

GPT-2 Medium
& 61.3$\pm$1.2
& 67.7$\pm$3.1
& +6.4 \\

GPT-2 Large
& 69.3$\pm$1.2
& 58.3$\pm$7.4
& $-$11.0 \\

OPT-1.3B
& 35.0$\pm$2.7
& 13.0$\pm$3.6
& $-$22.0 \\
\bottomrule
\end{tabular}
\end{table}

\paragraph{More trainable parameters do not uniformly increase recovery.}
As shown in Figure~\ref{fig:appendix_lora_fullft} and
Table~\ref{tab:lora_fullft}, full fine-tuning improves mean recovery for
GPT-2 Small and Medium, but is substantially worse for GPT-2 Large and
OPT-1.3B and exhibits larger variance in several settings. Although this
historical experiment is not directly comparable to the current main
protocol, it shows that lightweight LoRA adaptation is sufficient to produce
substantial recovery and that increasing the number of trainable parameters
does not guarantee stronger recovery.

\subsection{Random-Seed Consistency}
\label{app:seed_stability}
The headline Matched Random comparison uses the same three pre-specified
seeds, 42, 43, and 44, across model settings. On GPT-2 Medium, the
corresponding recovery rates are 55.0\%, 51.0\%, and 53.0\%, yielding
$53.0\pm2.0\%$. These are the values used in the main-paper comparison.
Full per-seed results for all six models are reported in
Table~\ref{tab:appendix_per_seed}.

\subsection{Genuine-Private and ICL Baselines}
\label{app:additional_baselines}

\paragraph{Genuine-Private LoRA.}
We compare \method with a stronger-supervision reference that replaces
model-generated supervision with 50 genuine private-reference associations.
For GPT-2 Small, Medium, and Large, OPT-1.3B, and Qwen3-1.7B, this baseline
uses the same target model and LoRA adaptation budget as \method .
It is therefore a Janus-style genuine-private reference rather than an exact
reproduction of the original Janus protocol.

\begin{table}[t]
\centering
\small
\caption{Comparison with genuine-private supervision. Values are target
recovery percentages; stochastic results report mean $\pm$ sample standard
deviation over three seeds.}
\label{tab:genuine_private}
\begin{tabular}{lccc}
\toprule
Model & Hidden & \method & Genuine-Private LoRA \\
\midrule
GPT-2 Small  & 23.0 & $40.3\pm1.5$ & $46.3\pm0.6$ \\
GPT-2 Medium & 42.0 & $63.0\pm2.6$ & $66.3\pm1.2$ \\
GPT-2 Large  & 55.0 & $72.0\pm1.0$ & $77.0\pm1.0$ \\
OPT-1.3B     & 21.0 & $30.0\pm2.7$ & $69.7\pm1.2$ \\
Qwen3-1.7B   & 52.0 & $58.0\pm1.7$ & $74.0\pm1.0$ \\
\bottomrule
\end{tabular}
\end{table}

Genuine-private supervision is stronger overall, but the gap is
model-dependent. On GPT-2 Medium and Large, Genuine-Private LoRA exceeds
\method by only 3.3 and 5.0 percentage points, respectively, whereas the gap is
substantially larger on OPT-1.3B and Qwen3-1.7B.

For Qwen3-4B, a matched Genuine-Private LoRA run is not available.
The 66.0\% private-supervision value shown in the main figure is included
only as an unmatched descriptive reference and should not be interpreted as
a matched result under the current protocol.

\paragraph{Janus-style in-context learning.}
We additionally evaluate whether genuine private examples can improve recovery
without parameter updates. Following the comparison motivated by
Janus~\citep{chen2024janus}, we use genuine private associations as
in-context demonstrations and evaluate 1-, 5-, 10-, and 20-shot prompting on
GPT-2 Medium and Large.

\begin{table}[t]
\centering
\small
\caption{Janus-style in-context learning comparison. Values are recovery
percentages; stochastic results report mean $\pm$ sample standard deviation
over three seeds.}
\label{tab:icl_baseline}
\begin{tabular}{lcc}
\toprule
Method & GPT-2 Medium & GPT-2 Large \\
\midrule
Hidden              & 42.0 & 55.0 \\
ICL (1-shot)         & $22.0\pm1.7$ & $33.0\pm2.0$ \\
ICL (5-shot)         & $6.7\pm3.5$  & $13.3\pm11.5$ \\
ICL (10-shot)        & $10.0\pm5.3$ & $13.3\pm5.5$ \\
ICL (20-shot)        & $10.7\pm3.8$ & $2.7\pm1.2$ \\
\midrule
\method                & $63.0\pm2.6$ & $72.0\pm1.0$ \\
Genuine-Private LoRA & $66.3\pm1.2$ & $77.0\pm1.0$ \\
\bottomrule
\end{tabular}
\end{table}

Under this matched prompting protocol, ICL does not improve recovery over the
unadapted target model on either model. Increasing the number of
demonstrations also does not monotonically improve recovery. We therefore use
ICL as an additional parameter-free reference rather than as a substitute for
the adaptation baselines.

\section{Post-hoc Analysis}
\label{app:posthoc}

This section provides additional post-hoc analyses of the supervision selected
during recovery. These experiments are intended to characterize empirical
relationships between candidate compatibility, gradient direction, and
recovery performance. They are diagnostic analyses rather than components of
the \method attack itself: target gradients are never available during candidate
generation, selection, or adaptation.

\subsection{Candidate Compatibility and Gradient Alignment}
\label{app:gradient_analysis}

We first examine whether candidates that appear more locally compatible with
the target model also tend to induce update directions that are more
aligned with target recovery.

For GPT-2 Medium, we analyze 100 candidate examples and compute two quantities
post hoc. The first is a candidate-level local compatibility score derived from
the frozen checkpoint. The second is the cosine similarity between the
candidate-induced gradient and a target-recovery gradient computed only for
analysis after the experiment has been completed.

\begin{figure}[t]
    \centering
    \includegraphics[width=0.7\columnwidth]
    {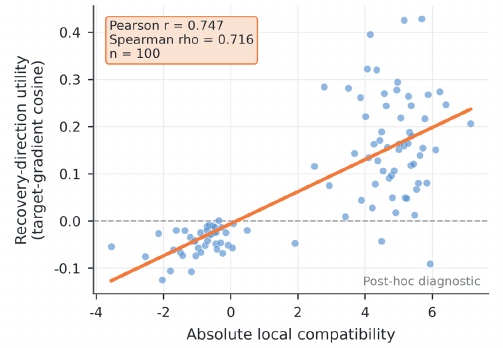}
\caption{Candidate compatibility versus target-gradient alignment on GPT-2 Medium.}
    \label{fig:gradient_analysis}
\end{figure}

As shown in Figure~\ref{fig:gradient_analysis}, candidate compatibility is
positively associated with target-gradient alignment. We observe a Pearson
correlation of
\[
r = 0.747
\]
and a Spearman rank correlation of
\[
\rho = 0.716.
\]
These results indicate that candidates scoring more favorably under the
post-hoc compatibility measure also tend to induce update directions that are
more aligned with the target-recovery direction.

\paragraph{Interpretation.}
This correlation provides evidence that local model compatibility is related
to the usefulness of candidate supervision for recovery. However, it does not
establish a causal mechanism, and the target gradient used in this analysis is
not available to the attack. The main \method selector remains the answer-token
negative log-likelihood computed solely from the frozen target model.

\subsection{Set-Level Gradient Comparison}
\label{app:set_gradient}

We next compare the aggregate gradient direction induced by several candidate
sets. For each set, we compute the cosine similarity between the update
direction produced by the selected supervision and the target-recovery
gradient.

This analysis was conducted using an earlier selector-analysis protocol and is
reported only as a post-hoc diagnostic.

\begin{table}[t]
\centering
\small
\caption{
Post-hoc set-level gradient cosine similarity on GPT-2 Medium.
Higher values indicate stronger alignment with the target-recovery direction.
The genuine-private set is shown only as a reference and is unavailable under
the \method threat model.
}
\label{tab:set_gradient_cosine}
\begin{tabular}{lc}
\toprule
Supervision set & Gradient cosine \\
\midrule
Bottom-ranked candidates & $-0.077$ \\
Random candidates        & $0.029$ \\
Historical \method-Abs selection & $0.633$ \\
Genuine-private reference & $0.805$ \\
\bottomrule
\end{tabular}
\end{table}

As shown in Table~\ref{tab:set_gradient_cosine}, the aggregate direction induced by randomly selected candidates is nearly orthogonal to the target-recovery gradient, while the historical compatibility-based selection produces substantially stronger positive alignment. The genuine-private reference provides an upper reference point in
this diagnostic.

These values are consistent with the hypothesis that candidate selection can
concentrate supervision whose induced updates are more compatible with target
recovery. We emphasize that this comparison is post hoc and does not imply
that \method estimates or optimizes target gradients during the attack.

\subsection{Additional Selector Analysis}
\label{app:selector_analysis}

Before adopting the simple Original-Loss selector used in the main paper, we
evaluated several alternative candidate-ranking strategies. These experiments
belong to an earlier selector-development protocol and are included to document
the robustness of the broader candidate-selection idea rather than to redefine
the current method.

The evaluated alternatives include a relative local-loss-gap score, a
local-basin score, and an automatically searched selector configuration.
Their best observed recovery rates on GPT-2 Medium are summarized in
Table~\ref{tab:additional_selectors}.

\begin{table}[t]
\centering
\small
\caption{
Historical selector-development results on GPT-2 Medium.
These values come from an earlier protocol and are not directly comparable to
the current one-round Original-Loss headline results.
}
\label{tab:additional_selectors}
\begin{tabular}{lc}
\toprule
Selector variant & Best recovery (\%) \\
\midrule
Relative local-gap & 64.7 \\
Local-basin        & 64.4 \\
Automatic selector search & 64.0 \\
\bottomrule
\end{tabular}
\end{table}

The more elaborate selectors do not provide a consistent advantage over the
best manually specified alternatives. In particular, automatic selector search
does not reliably improve upon simpler compatibility-based scores.

These observations motivated a simpler presentation in the final paper:
the headline method uses frozen-model answer-token loss, which is easy to
define, inexpensive to compute, and does not require perturbation-based
statistics or target-side supervision. The historical selector experiments
are retained here only as evidence that the recovery phenomenon is not tied to
a single highly tuned ranking heuristic.
\section{Boundary Experiments}
\label{app:boundaries}

The main experiments establish that model-generated supervision can increase
post-adaptation recovery in structured person-email tasks. We next examine
where this effect weakens or disappears. These experiments vary the amount and
source of exploitable structure in the underlying associations.

Unless otherwise noted, the experiments in this section were conducted under
earlier \method-Abs protocols and use independently trained checkpoints.
They are therefore intended as descriptive boundary analyses rather than
strictly matched extensions of the current one-round Original-Loss protocol.

\subsection{Opaque Random Bindings}
\label{app:opaque}

A potential alternative explanation for recovery on person-email tasks is
that the model can exploit lexical or compositional regularities between a
person's identity and the corresponding email address. To remove these cues,
we construct a non-compositional benchmark in which synthetically generated
identities are independently paired with random 12-character local parts under
the reserved \texttt{privacy.invalid} domain.

The abstract task format remains person-to-email, but the target answer has
no designed lexical relationship with the entity name. As summarized in
Tables~\ref{tab:structured_vs_opaque} and~\ref{tab:opaque_boundary}, the previously exposed
opaque condition drops from 95.0\% recovery in the original model to 1.0\% after non-target continual training and reaches only 2.0\% after ReGap-Abs, while the matched target-unexposed branch remains at 0\%.

\begin{table}[t]
\centering
\small
\caption{Boundary comparison between structured associations and opaque random
bindings under the historical three-round \method-Abs protocol.}
\label{tab:structured_vs_opaque}
\begin{tabular}{lc}
\toprule
Association setting & Recovery rate (\%) \\
\midrule
Structured associations & 61.3 \\
Opaque random bindings & 2.0 \\
\bottomrule
\end{tabular}
\end{table}

\begin{table}[t]
\centering
\small
\caption{
Recovery on opaque non-compositional bindings.
The experiment uses the historical three-round \method-Abs protocol.
Values are percentages.
}
\label{tab:opaque_boundary}
\begin{tabular}{lccc}
\toprule
Exposure condition & Memorized & Hidden & \method \\
\midrule
Previously Exposed & 95.0 & 1.0 & 2.0 \\
Target-Unexposed   & -   & 0.0 & 0.0 \\
\bottomrule
\end{tabular}
\end{table}

No exact target answer appears in either the generated candidate pools or the
selected pseudo-supervision in this experiment. Thus, the weak recovery cannot
be attributed to direct replay of correct target answers.

\paragraph{Opaque bindings define an empirical boundary.}
Previously memorized arbitrary bindings show almost no answer-free recovery
after direct accessibility is reduced. This result contrasts with the much
larger gains observed in structured person-email settings and is consistent
with exploitable association structure playing an important role in recovery.
It does not establish that structure is the only factor determining
recoverability, since the opaque branch uses its own independently trained
checkpoint and historical adaptation protocol.

\subsection{Controlled Enron Associations}
\label{app:controlled_enron}

The main benchmark derives person identifiers mechanically from email
local-parts, creating an explicit lexical relationship between query identities
and target answers. To examine recovery when names and addresses originate
from separate fields, we additionally construct a Controlled Enron benchmark
from structured message metadata.

Sender display names or \texttt{X-From} fields are paired with the
corresponding sender email addresses from the same message.
After filtering, the dataset contains 757 clean one-to-one person-email
associations.

Under the historical three-round \method-Abs protocol, recovery increases from
27.0\% at the target model to
77.7$\pm$1.5\%.
Random model-generated supervision also produces substantial recovery at
66.7$\pm$2.3\%, while Original-Loss selection reaches
75.7$\pm$2.1\%.
A genuine-private reference reaches 98.0\%.

\begin{table*}[t]
\centering
\small
\caption{
Controlled Enron and Non-Enron recovery under the historical three-round
protocol. The two settings use independently constructed checkpoints and are
reported as descriptive boundary evidence rather than as a matched causal
ablation. Values are percentages and report mean $\pm$ sample standard
deviation where applicable.
}
\label{tab:controlled_enron}
\begin{tabular}{lccccc}
\toprule
Setting & Hidden & Random & Original Loss & \method-Abs &
Genuine-Private \\
\midrule
Controlled Enron
& 27.0
& 66.7$\pm$2.3
& 75.7$\pm$2.1
& 77.7$\pm$1.5
& 98.0$\pm$0.0 \\

Non-Enron
& 46.0
& 50.3$\pm$2.1
& 42.0$\pm$9.9
& 46.7$\pm$3.2
& 80.7$\pm$1.2 \\
\bottomrule
\end{tabular}
\end{table*}

The Controlled Enron experiment suggests that the observed recovery is not
restricted to associations constructed by mechanically converting email
local-parts into person identifiers. However, it does not establish that all
real-world person-email distributions support similarly strong recovery.

\paragraph{Recovery is not explained solely by direct candidate replay.}
We further examine targets whose correct answers are absent from their own
first-round candidate pools. On this subset, recovery increases from 0.7\% at
the target model to 66.0\% after the third \method round.
Across all three seeds, none of these target answers appears in another
target's first-round candidate pool or selected pseudo-supervision.

This audit rules out direct first-round exact-answer replay for this subset.
It does not, however, establish a particular internal memory mechanism or rule
out the use of broader distributional regularities.

\subsection{Non-Enron-Domain Associations}
\label{app:non_enron}

The Controlled Enron data contain a strong corporate-domain regularity:
many target addresses belong to \texttt{enron.com}.
We therefore construct an additional subset that excludes addresses from
\texttt{enron.com} and its subdomains.
The resulting dataset contains 714 external-domain person-email associations.

Removing the dominant corporate-domain regularity sharply attenuates the
answer-free recovery effect.
As shown in Table~\ref{tab:controlled_enron}, the target model achieves
46.0\% recovery, while \method-Abs reaches
46.7$\pm$3.2\%, corresponding to only a small mean increase.
Matched Random Selection reaches 50.3$\pm$2.1\%, and Original-Loss selection
reaches 42.0$\pm$9.9\%.

In contrast, the genuine-private reference reaches
80.7$\pm$1.2\%, showing that substantial adaptation headroom remains even when
answer-free model-generated supervision provides little additional recovery.

\paragraph{The benefit of model-generated supervision is distribution dependent.}
The contrast between Controlled Enron and Non-Enron suggests that the strength
of answer-free recovery depends substantially on the statistical structure
available in the association distribution.
Removing the dominant domain pattern weakens the effect despite retaining the
same broad person-email task format.

We interpret this result conservatively.
The Controlled Enron and Non-Enron checkpoints are constructed independently,
and the comparison is therefore not a strictly matched intervention on a
single trained model.
Nevertheless, together with the opaque-binding experiment, it provides
descriptive evidence that \method does not act as a generic mechanism for
recovering arbitrary hidden associations.
\subsection{Explicit NPO Unlearning Stress Test}
\label{app:npo_defense}

Our main target models reduce direct target accessibility through non-target
continual training. We additionally consider a distinct setting in which
GPT-2 Medium is first modified by Negative Preference Optimization (NPO),
producing two defended checkpoints with different levels of direct target
recoverability, denoted Moderate and Strong.

Unlike the historical experiments in Sections~\ref{app:opaque}--%
\ref{app:non_enron}, the recovery stage of this stress test uses the current
one-round Original-Loss R1 protocol. Starting independently from each NPO
checkpoint, we compare the unadapted model with Matched Random and \method
under the same one-round adaptation procedure used in the main experiments.
The NPO checkpoints themselves are produced by an explicit defense objective
rather than by the non-target continual-training procedure used to construct
the main target models. We therefore treat this experiment as a boundary
stress test under a different accessibility-reduction mechanism, rather than
as a strictly matched causal comparison with the main setting.

\begin{table}[t]
\centering
\small
\caption{Recovery after explicit NPO unlearning on GPT-2 Medium. Values are
percentages; stochastic results report mean $\pm$ sample standard deviation
over three seeds.}
\label{tab:npo_defense}
\begin{tabular}{lccc}
\toprule
NPO setting & Before adaptation & Matched Random & \method \\
\midrule
Moderate & 71.0 & $72.0\pm0.0$ & $71.3\pm0.6$ \\
Strong   & 32.0 & $32.0\pm0.0$ & $32.0\pm0.0$ \\
\bottomrule
\end{tabular}
\end{table}

As shown in Table~\ref{tab:npo_defense}, \method provides almost no additional
recovery from either NPO checkpoint. Under the Moderate setting, recovery
changes from 71.0\% before adaptation to $71.3\pm0.6\%$ after \method.
Under the Strong setting, recovery remains unchanged at 32.0\%. Matched Random
exhibits the same qualitative behavior.

These results show that the substantial recovery observed after ordinary
non-target continual training does not automatically extend to checkpoints
produced by explicit NPO. Because the NPO checkpoints arise from a different
accessibility-reduction mechanism, this experiment should be interpreted as a
boundary stress test rather than as a matched comparison between continual
training and NPO. It also covers only one model and one unlearning method and
therefore does not establish general robustness to unlearning defenses.

\section{Current-Protocol Boundary Controls}
\label{app:current_boundary}

To further characterize the boundary conditions of \method under the
current one-round Original-Loss protocol, we evaluate two additional
controls. Unlike the historical boundary experiments in
Appendix~\ref{app:boundaries}, these experiments follow the same evaluation
protocol as the main study: one Generate-Select-Adapt round, Original-Loss
candidate selection, and three random seeds.

\subsection{Opaque Random Binding under Current R1}
\label{app:opaque_current}

We evaluate \method on an Opaque random-binding benchmark where the association
structure is preserved but semantic information about the target mapping is
removed. This setting tests whether recovery can occur when the model itself
cannot generate informative target-specific supervision.

\begin{table}[t]
\centering
\caption{Current-protocol Opaque random-binding control.}
\label{tab:current_opaque}
\begin{tabular}{lccc}
\toprule
Setting & Hidden & Random & \method \\
\midrule
GPT-2 Medium
& 1.0 & 1.7$\pm$0.6 & 1.0$\pm$0.0 \\
\bottomrule
\end{tabular}
\end{table}

Across all three seeds, no exact target answer appears in either the generated candidate pools or the selected supervision examples. As shown in Table~\ref{tab:current_opaque},
\method does not improve over the target model in this setting. This result suggests
that adaptation alone is insufficient: recovery requires model-generated
supervision signals that preserve target-specific information.

\subsection{Controlled Enron under Current R1}
\label{app:controlled_current}

We additionally evaluate a Controlled Enron association benchmark under the
current protocol. This experiment uses a GPT-2 Large post-hoc target model; therefore, it is treated as a supplementary boundary analysis
rather than a strict matched-scale comparison with the main experiments.

\begin{table}[t]
\centering
\caption{Current-protocol Controlled Enron boundary control.}
\label{tab:controlled_current}
\begin{tabular}{lccc}
\toprule
Setting & Hidden & Random & \method \\
\midrule
GPT-2 Large
& 29.0 & 36.0$\pm$1.7 & 84.0$\pm$2.6 \\
\bottomrule
\end{tabular}
\end{table}

As shown in Table~\ref{tab:controlled_current}, \method substantially increases recovery in this setting. However, the strong recovery is accompanied by target-answer replay in the generated supervision. Across seeds, each candidate pool contains
49-51 exact target answers, and the selected supervision contains
25-27 target answers. Therefore, this result should be interpreted as
recovery enabled by model-generated target supervision rather than
fully answer-free recovery.
\section{Leakage and Reproducibility Audits}
\label{app:audits}

This section summarizes the audits used to verify that the reported recovery
results are not artifacts of target leakage, accidental answer replay,
checkpoint mismatch, or seed reuse. We distinguish audits of the main
target-identity-aware setting from those of the Unseen-Identity Control,
because exact target-answer replay is ruled out only in the latter.

\subsection{Identity-Overlap Audit}
\label{app:identity_audit}

For the Unseen-Identity Control, adaptation identities are generated
independently from the privacy benchmark.
We explicitly compare them against all identity sets used in the experiment:
the private-reference set, target set, non-target evaluation set, and the
complete memorized privacy pool.

Thus, the recovery observed in the Unseen-Identity Control cannot be
attributed to reusing target or memorized identities during adaptation.
The control nevertheless preserves the person-email task structure and
expected answer format, and should therefore not be interpreted as a
task-unaware setting.

\subsection{Exact Target-Answer Replay Audit}
\label{app:replay_audit}

We next audit whether the generated or selected adaptation examples contain
exact target answers.

For the Unseen-Identity Control, no exact target answer is found in either
the generated candidate pools or the selected adaptation examples across all
three evaluated seeds for GPT-2 Medium and GPT-2 Large. Table~\ref{tab:unseen_identity_audit} summarizes these overlap and replay checks,
all of which are zero across the evaluated seeds.

\begin{table}[t]
\centering
\caption{Leakage audit for the Unseen-Identity Control. 
All overlap and replay counts are zero across evaluated seeds.}
\label{tab:unseen_identity_audit}
\begin{tabular}{lc}
\toprule
Audit criterion & Result \\
\midrule
Private-reference identity overlap & 0 \\
Target identity overlap & 0 \\
Non-target evaluation identity overlap & 0 \\
All memorized privacy identity overlap & 0 \\
\midrule
Candidate-pool target-answer hits & 0  \\
Selected-example target-answer hits & 0 \\
\bottomrule
\end{tabular}
\end{table}

This result rules out exact target-answer replay as an explanation for the
Unseen-Identity recovery effect.

Importantly, this conclusion does not extend to the primary
target-identity-aware setting. In that setting, generated candidates can
occasionally match target answers exactly, and selection can concentrate some
of these candidates. We therefore do not claim that exact replay is absent
from the main target-aware attack. Instead, the Unseen-Identity Control
provides the cleaner evidence that substantial recovery can occur even when
exact target answers never appear in the generated or selected supervision.

\subsection{Frozen-Scorer and Fresh-Adapter Verification}
\label{app:protocol_audit}

The main \method protocol separates candidate generation and scoring from
adaptation.

For every main experiment:

\begin{itemize}
    \item candidate generation uses the fixed target model;
    \item candidate scoring uses the same frozen target model;
    \item target answers are not used during generation or scoring;
    \item each adaptation run initializes a fresh LoRA adapter from the same
    target model;
    \item adapters are not continued across random seeds; and
    \item the final adapted model is evaluated only after training is complete.
\end{itemize}

The matched Random Selection baseline follows the same protocol except for the
candidate-selection rule. It uses the same generated candidate pool, the same
number of selected examples, the same target model, and the same LoRA
optimization configuration.

These checks ensure that the \method-Random comparison primarily changes which generated examples are selected for adaptation, while holding the training budget and model checkpoint fixed.

\subsection{Seed and Protocol Consistency}
\label{app:seed_audit}

All headline stochastic results use seeds 42, 43, and 44 under the current
one-round Original-Loss protocol. As shown in Table~\ref{tab:appendix_per_seed}, for GPT-2
Medium, the matched Random recoveries are 55.0\%, 51.0\%, and 53.0\%,
giving $53.0\pm2.0\%$. The corresponding \method recoveries are
60.0\%, 64.0\%, and 65.0\%, giving $63.0\pm2.6\%$.

We additionally verify the current prior-exposure experiment using the same
100 target associations and the same one-round Original-Loss protocol.
As shown in Table~\ref{tab:prior_exposure_reproduction}, on the previously exposed target model,
\method increases recovery from 42.0\% to $63.0\pm2.6\%$, while Matched
Random reaches $53.0\pm2.0\%$. On the target-unexposed control, both remain
at 13.0\% for all three seeds.

\begin{table}[t]
\centering
\small
\caption{
Reproduction audit for the current prior-exposure control.
All conditions use the same 100 target associations and the same evaluation
protocol. Values are percentages.
}
\label{tab:prior_exposure_reproduction}
\begin{tabular}{lccc}
\toprule
Model checkpoint & Baseline & \method R1 & Random R1 \\
\midrule
Previously Exposed
& 42.0
& 63.0$\pm$2.6
& 53.0$\pm$2.0 \\

Target-Unexposed Control
& 13.0
& 13.0$\pm$0.0
& 13.0$\pm$0.0 \\
\bottomrule
\end{tabular}
\end{table}

The reproduced aggregate results exactly match the corresponding main-paper
numbers. The current prior-exposure control therefore uses the same
Original-Loss R1 adapters as the main GPT-2 Medium experiment rather than the
historical \method-Norm adapters used in earlier development runs.

\subsection{Result Provenance}
\label{app:provenance}

All headline results were verified against the corresponding audited
current-protocol experiment outputs before inclusion in the paper. The
six-model main results use the current one-round Original-Loss protocol,
with the evaluation procedure and random seeds specified in
Appendix~\ref{app:implementation}.

Historical experiments using earlier selectors, multi-round configurations,
or independently constructed checkpoints are reported only where explicitly
labeled as historical or diagnostic. These results are kept separate from the
current-protocol headline results and are not used to replace or augment the
primary comparisons.

Together with the identity-overlap, exact-answer replay, frozen-scorer,
fresh-adapter, seed, and checkpoint audits above, this protocol separation
provides additional checks against accidental leakage, mismatched adaptation
budgets, adapter reuse across seeds, and mixing results obtained under
different experimental configurations.

\section{Theoretical Analysis and Proofs}
\label{app:theory}

\subsection{Analytical Setup}
\label{app:theory_setup}

All gradients are taken in the fixed trainable coordinates of $M(\phi)$.
Candidate texts, prompts, tokenizations, and selection decisions are held
fixed when differentiating a training loss. In particular, we do not
differentiate through generation or selection. The target examples are
fixed analytical objects and are unavailable to the attack.

For a differentiable function $F$, $\beta$-smoothness on a segment means
$\|\nabla F(u)-\nabla F(v)\|_2\le\beta\|u-v\|_2$ for points $u,v$
on that segment. Results involving expectations assume the stated first
or second moments exist. Bounds on random updates additionally require a
common smoothness constant on the corresponding update segments almost
surely. No convexity assumption is used.

\subsection{Proof of Theorem~\ref{thm:local_recovery}}
\label{app:proof_local}

\begin{proof}
By the fundamental theorem of calculus,
\begin{align}
L_T(\delta)-L_T(0)
&=\int_0^1\langle\nabla L_T(t\delta),\delta\rangle\,dt\\
&=\langle g_T,\delta\rangle
+\int_0^1\langle\nabla L_T(t\delta)-g_T,\delta\rangle\,dt.
\end{align}
Cauchy--Schwarz and smoothness bound the integrand in the second term
by $\beta t\|\delta\|_2^2$. Integration gives
Equation~\ref{eq:general_update_bound}. Substituting $\delta=-\eta g_S$
and using
\begin{equation}
\langle g_T,g_S\rangle
=\left\langle\frac{1}{N}\sum_{i=1}^{N}g(t_i),
\frac{1}{|S|}\sum_{z\in S}g(z)\right\rangle
=\frac{1}{N|S|}\sum_{i=1}^{N}\sum_{z\in S}\kappa_h(t_i,z)
\end{equation}
gives Equation~\ref{eq:local_recovery_bound}. Since $\eta>0$,
Equation~\ref{eq:recovery_condition} makes its right-hand side strictly
smaller than $L_T(0)$.
\end{proof}

\subsection{Actual Updates and Multi-Step Fine-Tuning}
\label{app:multistep_bound}

\begin{corollary}[Trajectory-wise target-loss bound]
\label{cor:trajectory_bound}
Let $\phi_{h+1}=\phi_h+\delta_h$ for $h=0,\ldots,H-1$.
If $L_T$ is $\beta_h$-smooth on each update segment, then
\begin{equation}
L_T(\phi_H)-L_T(\phi_0)
\le\sum_{h=0}^{H-1}
\left(
\langle\nabla L_T(\phi_h),\delta_h\rangle
+\frac{\beta_h}{2}\|\delta_h\|_2^2
\right).
\label{eq:trajectory_bound}
\end{equation}
\end{corollary}

\begin{proof}
Apply the integral argument above at $\phi_h$ with displacement $\delta_h$
and sum the resulting inequalities. The left-hand sides telescope.
\end{proof}

This statement uses the actual optimizer displacements, including any
momentum, preconditioning, or weight decay. For the preconditioned update
$\delta_h=-\eta_h P_h\nabla L_S(\phi_h)$, the first-order term is
\[
-\eta_h\langle\nabla L_T(\phi_h),P_h\nabla L_S(\phi_h)\rangle.
\]
It need not have the same sign as the unpreconditioned gradient inner
product. Initial alignment therefore does not imply persistent alignment
along an arbitrary optimization trajectory.

\subsection{Shared Gradient Structure Without Identity Overlap}
\label{app:shared_structure}

\begin{assumption}[Shared task-gradient structure]
\label{ass:shared_structure}
Let $U\in\mathbb R^{d\times r}$ satisfy $U^\top U=I_r$. For a fixed
target $t_i$ and $z\sim Q$, suppose
\begin{equation}
g(t_i)=Ua_i+r_i,
\qquad
g(z)=Ub(z)+\xi(z),
\label{eq:shared_decomposition}
\end{equation}
where $U^\top r_i=0$ and $U^\top\xi(z)=0$. Assume
$\mathbb E_{z\sim Q}\|g(z)\|_2<\infty$, set
$\mu_Q=\mathbb E_{z\sim Q}[b(z)]$, and suppose
$|\mathbb E_{z\sim Q}\langle r_i,\xi(z)\rangle|\le\epsilon$
for some $\epsilon\ge0$.
\end{assumption}

\begin{proposition}[Identity-disjoint positive alignment]
\label{prop:shared_structure}
Under Assumption~\ref{ass:shared_structure}, if
$a_i^\top\mu_Q\ge\gamma$, then
\begin{equation}
\sigma_i(Q)\ge\gamma-\epsilon.
\label{eq:shared_alignment_bound}
\end{equation}
In particular, $\gamma>\epsilon$ implies positive expected alignment.
The result also holds when $Q$ is supported entirely on examples whose
identities and answer strings differ from those of $t_i$.
\end{proposition}

\begin{proof}
Orthogonality eliminates the cross terms, yielding
\begin{equation}
\kappa_h(t_i,z)
=a_i^\top b(z)+\langle r_i,\xi(z)\rangle.
\end{equation}
Taking expectations gives
$\sigma_i(Q)=a_i^\top\mu_Q+
\mathbb E\langle r_i,\xi(z)\rangle\ge\gamma-\epsilon$.
Neither expansion uses identity or answer overlap.
\end{proof}

The projection decomposition can be formed for any fixed $U$; its
substantive conditions are positive shared alignment and a controlled
residual contribution. Identity disjointness alone establishes neither.
Likewise, removing compositional structure does not mathematically force
$\mu_Q=0$, and the proposition does not compare models with different
exposure histories. These are separate empirical mechanism questions.

\subsection{Self-Sampling and Answer-Length Weighting}
\label{app:self_sampling}

\begin{proposition}[Matched self-sampling nullity]
\label{prop:self_sampling_nullity}
Fix $q$. Let $P_\phi(\cdot\mid q)$ be a differentiable normalized
distribution on a fixed countable support of complete answers, with
positive probabilities on that support near $\phi=0$. Assume
differentiation can be interchanged with summation. If
$Y\sim P_0(\cdot\mid q)$ and
$\widetilde\ell_Y(\phi)=-\log P_\phi(Y\mid q)$, then
\begin{equation}
\mathbb E_{Y\sim P_0(\cdot\mid q)}
[\nabla_\phi\widetilde\ell_Y(0)]=0.
\label{eq:self_sampling_nullity}
\end{equation}
\end{proposition}

\begin{proof}
Holding the sampling distribution fixed at $P_0$,
\begin{align}
\mathbb E[\nabla_\phi\widetilde\ell_Y(0)]
&=-\sum_y P_0(y\mid q)
\frac{\left.\nabla_\phi P_\phi(y\mid q)\right|_{\phi=0}}
{P_0(y\mid q)}\\
&=-\left.\nabla_\phi\sum_yP_\phi(y\mid q)\right|_{\phi=0}
=-\nabla_\phi 1=0.
\end{align}
\end{proof}

The distribution and likelihood must describe the same complete-answer
event, including termination or any explicitly modeled truncation.
The proposition does not automatically apply to an objective that omits
part of that event or reweights answers by length.

To make the weighting issue explicit, let
$\widetilde g(Y)=\nabla\widetilde\ell_Y(0)$ and consider the matched but
length-weighted objective $\widetilde\ell_Y(\phi)/n(Y)$ with
$n(Y)\ge1$. When the relevant moments exist,
\begin{equation}
\mathbb E\!\left[\frac{\widetilde g(Y)}{n(Y)}\right]
=\mathbb E\!\left[
\left(\frac{1}{n(Y)}-\mathbb E\frac{1}{n(Y)}\right)
\widetilde g(Y)\right],
\label{eq:length_weighted_sampling}
\end{equation}
which need not vanish. This follows by subtracting
$\mathbb E[1/n(Y)]\mathbb E[\widetilde g(Y)]=0$.
For example, let two complete answers have lengths $n_a\ne n_b$ and
probabilities $P_\phi(y_a)=\operatorname{sigmoid}(\phi)$ and
$P_\phi(y_b)=1-\operatorname{sigmoid}(\phi)$. At $\phi=0$, their
unweighted gradients are $-1/2$ and $1/2$, but their weighted expectation
is $\tfrac14(1/n_b-1/n_a)\ne0$.
Neither a nonzero weighted expectation nor biased generation establishes
positive alignment with the target gradient.

\subsection{Proof of Proposition~\ref{prop:topk_selector}}
\label{app:selector_proof}

\begin{proof}
Condition on $\mathbf s$. A score-only selector, possibly using independent
randomization, has conditional inclusion probabilities
$\pi_j(\mathbf s)\in[0,1]$ with $\sum_{j=1}^{m}\pi_j=K$.
Independence of its randomization from the candidates implies
\begin{equation}
\mathbb E[\langle g_T,g_S\rangle\mid\mathbf s]
=\frac{1}{K}\sum_{j=1}^{m}\pi_j(\mathbf s)a_j(\mathbf s).
\label{eq:selector_conditional_objective}
\end{equation}
Relabel the indices so that $s_h(z_1)\ge\cdots\ge s_h(z_m)$.
Assumption~\ref{ass:support_alignment} gives $a_1\ge\cdots\ge a_m$.
Writing $a_j=a_j(\mathbf s)$, we have
\begin{align}
\sum_{j=1}^{K}a_j-\sum_{j=1}^{m}\pi_j a_j
&=\sum_{j\le K}(1-\pi_j)(a_j-a_K)
+\sum_{j>K}\pi_j(a_K-a_j)
\ge0.
\end{align}
The equality uses $\sum_j\pi_j=K$. Top-$K$ selection attains the
upper bound by setting $\pi_j=1$ for $j\le K$ and zero otherwise.
Equal support scores have equal conditional expected alignments by the
assumption, so the tie-breaking rule does not affect optimality.
\end{proof}

No independence between pool candidates is needed: the assumption
conditions on the entire score vector. For random pool size, apply the
argument conditional on that size, with the corresponding fixed
$K=\min(K_0,m)$.

\subsection{From Conditional Ranking to a Finite-Step Loss Bound}
\label{app:selector_bound}

Let $S_K$ be the support-selected set, and define
\begin{equation}
A_K(\mathbf s)=\mathbb E[\langle g_T,g_{S_K}\rangle\mid\mathbf s],
\qquad
B_K(\mathbf s)=\mathbb E[\|g_{S_K}\|_2^2\mid\mathbf s].
\label{eq:conditional_alignment_moments}
\end{equation}

\begin{corollary}[Conditional expected target-loss decrease]
\label{cor:conditional_descent}
Fix $\eta>0$. Suppose $B_K(\mathbf s)<\infty$ and a common
$\beta>0$ bounds the smoothness of $L_T$ on the segments
$[0,-\eta g_{S_K}]$ almost surely. Then
\begin{equation}
\mathbb E[L_T(-\eta g_{S_K})-L_T(0)\mid\mathbf s]
\le-\eta A_K(\mathbf s)
+\frac{\beta\eta^2}{2}B_K(\mathbf s).
\label{eq:conditional_selection_bound}
\end{equation}
In particular, the conditional expected target loss decreases whenever
$A_K(\mathbf s)>0$ and
\[
0<\eta<\frac{2A_K(\mathbf s)}{\beta B_K(\mathbf s)},
\]
provided the same segment-smoothness condition holds.
\end{corollary}

\begin{proof}
Apply Theorem~\ref{thm:local_recovery} to each realized selected set and
take conditional expectations. The stated moment and smoothness conditions
justify the bound. If $A_K>0$, then $B_K>0$: otherwise $g_{S_K}=0$
conditionally almost surely, which would imply $A_K=0$. The stated
learning-rate inequality makes the upper bound strictly negative.
\end{proof}

Proposition~\ref{prop:topk_selector} maximizes $A_K$, not the complete
right-hand side of Equation~\ref{eq:conditional_selection_bound}, because
$B_K$ also depends on selection. Moreover, support need not satisfy the
ordering assumption. For a scalar-logit binary model with $p(1)=0.9$
and target answer $0$, the cross-entropy gradient is $g(y)=0.9-y$.
Hence the candidate $1$ has higher support than $0$, but
$I_T(1)=0.9(-0.1)=-0.09$ and $I_T(0)=0.9^2=0.81$.
Thus, high support alone does not imply beneficial alignment.
The moments above are analytical quantities, not inputs to an attack-side
learning-rate or model-selection rule.

\subsection{Connecting Loss Changes to Decoded Recovery}
\label{app:margin_crossing}

Let $R_i(\phi)$ be the indicator in Equation~\ref{eq:recovery} for
$M(\phi)$. Under a fixed realization of the evaluation protocol,
\begin{equation}
\operatorname{Rec}(M(\delta))-\operatorname{Rec}(M_h)
=\frac{|\{i:R_i(0)=0,\ R_i(\delta)=1\}|
-|\{i:R_i(0)=1,\ R_i(\delta)=0\}|}{N}.
\label{eq:recovery_transition_identity}
\end{equation}
The identity follows by summing $R_i(\delta)-R_i(0)$, whose possible
values are $1,-1,0$. It also holds after taking expectations over
evaluation randomness.

A target-loss decrease alone does not determine these transitions.
For example, in a one-token task with target $a$, changing probabilities
from $(p_a,p_b,p_c)=(0.40,0.35,0.25)$ to $(0.41,0.45,0.14)$ lowers
$-\log p_a$ but changes the greedy answer from correct to incorrect.

The next lemma gives a sufficient condition for the specific case of
unmodified greedy decoding. It is not a characterization of arbitrary
substring recovery or a claim about other decoding algorithms.
Let $y_i=(y_{i,1},\ldots,y_{i,n_i})$ be the tokenization of a target
attribute with nonempty normalized text. For each position $t$ and
competitor token $v\ne y_{i,t}$ in a finite vocabulary $\mathcal V$, define
\begin{equation}
d_{i,t,v}(\phi)
=\log p_{M(\phi)}(y_{i,t}\mid q_i,y_{i,<t})
-\log p_{M(\phi)}(v\mid q_i,y_{i,<t}).
\label{eq:token_pair_margin}
\end{equation}

\begin{lemma}[Sufficient greedy-recovery condition]
\label{lem:greedy_recovery}
Suppose the generation budget is at least $n_i$, the decoder greedily
maximizes the unmodified next-token probabilities, and no external stopping
rule truncates a correct prefix before $n_i$ tokens. If every
$d_{i,t,v}$ is differentiable and $\Lambda_i$-smooth on $[0,\delta]$, and
\begin{equation}
\min_{\substack{1\le t\le n_i\\v\in\mathcal V,\ v\ne y_{i,t}}}
\left\{
d_{i,t,v}(0)+\langle\nabla d_{i,t,v}(0),\delta\rangle
-\frac{\Lambda_i}{2}\|\delta\|_2^2
\right\}>0,
\label{eq:greedy_margin_condition}
\end{equation}
then the generated continuation begins with $y_i$ and $R_i(\delta)=1$.
\end{lemma}

\begin{proof}
Applying the smoothness inequality to $-d_{i,t,v}$ gives
\begin{equation}
d_{i,t,v}(\delta)
\ge d_{i,t,v}(0)+\langle\nabla d_{i,t,v}(0),\delta\rangle
-\frac{\Lambda_i}{2}\|\delta\|_2^2>0.
\end{equation}
At the first position, the correct token strictly exceeds every competitor
in probability and is selected. Inductively, if the first $t-1$ tokens
equal $y_{i,<t}$, the decoder encounters exactly the prefix used in
Equation~\ref{eq:token_pair_margin}, and again selects $y_{i,t}$.
The budget and stopping assumptions allow this induction through $n_i$.
The resulting continuation therefore begins with the target attribute,
so its nonempty normalized string is included in the normalized output.
\end{proof}

Note this condition is only sufficient but not necessary: a target may also appear
later in the continuation. Together with
Equation~\ref{eq:recovery_transition_identity}, it implies a gain of at
least $|\mathcal J|/N$ if a set $\mathcal J$ of previously unrecovered
targets satisfies the condition and all previously recovered targets remain
recovered. Neither property follows solely from reducing the average
teacher-forced target loss. The recovery definition remains the normalized
substring metric in Equation~\ref{eq:recovery}; whole-field extraction
accuracy is a different metric.

\end{document}